\documentclass[a4paper, 11pt]{article}
\usepackage{graphicx}
\usepackage{amsmath, amssymb, amsthm}
\usepackage{algorithm}
\usepackage{algpseudocode}
\usepackage{overpic}
\usepackage{color}
\usepackage{mathrsfs}
\usepackage{multirow}
\usepackage[hang]{subfigure}

\newtheorem{theorem}{Theorem}
\newtheorem{lemma}[theorem]{Lemma}

\newcommand\bn{\boldsymbol{n}}

\newcommand\bv{\boldsymbol{v}}

\newcommand\bg{\boldsymbol{g}}

\newcommand\bbR{\mathbb{R}}
\newcommand\bbN{\mathbb{N}}

\newcommand\dd{\,\mathrm{d}}

\newcommand\mQ{\mathcal{Q}}
\newcommand\mM{\mathcal{M}}

\newcommand{\indexk}{k_1k_2k_3}
\newcommand{\indexi}{i_1i_2i_3}
\newcommand{\indexj}{j_1j_2j_3}

\newcommand{\indexl}{l_1l_2l_3}

\newcommand{\factorialk}{k_1!k_2!k_3!}

\newcommand\Talmi[4]{\left(\begin{array}{c|c}#1&#3\\#2&#4\end{array}\right)}
\numberwithin{equation}{section}
\newcommand{\imag}{\mathrm{i}}
\graphicspath{{images/}}

{\theoremstyle{remark} }

\DeclareMathOperator\diag{diag}

\title{Burnett Spectral Method for the Spatially Homogeneous Boltzmann Equation}

\author{Zhenning Cai\thanks{Department of Mathematics, National
    University of Singapore, Level 4, Block S17, 10 Lower Kent Ridge
    Road, Singapore 119076, email: {\tt matcz@nus.edu.sg}.},~~
Yuwei Fan\thanks{Department of Mathematics, Stanford University, Stanford,
    CA 94305, email: {\tt ywfan@stanford.edu}.},~~
Yanli Wang\thanks{Department of Engineering, Peking University, Beijing,
    China, 100871, email: {\tt wang\_yanli@pku.edu.cn}.}}

\begin{document}
\maketitle
\begin{abstract}
We develop a spectral method for the spatially homogeneous Boltzmann equation
using Burnett polynomials in the basis functions. Using the sparsity of the
coefficients in the expansion of the collision term, the computational cost is
reduced by one order of magnitude for general collision kernels and 
by two orders of magnitude for Maxwell molecules. The proposed method can couple seamlessly
with the BGK-type modelling techniques to make future applications affordable.
The implementation of the algorithm is discussed in detail, including a
numerical scheme to compute all the coefficients accurately, and the design of
the data structure to achieve high cache hit ratio. Numerical examples are
provided to demonstrate the accuracy and efficiency of our method.

\end{abstract}

\section{Introduction}
The Boltzmann equation possesses its unambiguous significance in the rarefied
gas dynamics. Using a velocity distribution function $f \in L^1(\mathbb{R}^3)$
to describe the statistical behavior of gas molecules, the Boltzmann equation
incorporates the transport and the collision of particles into a single
equation, which accurately models the gas flow from transitional to free
molecular regimes. By the molecular chaos assumption, the collision between
molecules gives the rate of change for the distribution function as follows:
\begin{equation} \label{eq:quad_col}
  \mQ[f,f](\bv) = \int_{\mathbb{R}^3}
  \int_{\bn \perp \bg} \int_0^{\pi} B(|\bg|,\chi)
    [f(\bv_1') f(\bv') - f(\bv_1) f(\bv)]
  \dd\chi \dd\bn \dd\bv_1,
\end{equation}
where $\bn$ is a unit vector and
\begin{align*}
\bg &= \bv - \bv_1, \\
\bv' &= \cos^2(\chi/2) \bv + \sin^2(\chi/2) \bv_1
  - |\bg| \cos(\chi/2) \sin(\chi/2) \bn, \\
\bv_1' &= \cos^2(\chi/2) \bv_1 + \sin^2(\chi/2) \bv
  + |\bg| \cos(\chi/2) \sin(\chi/2) \bn.
\end{align*}
The collisional kernel $B(\cdot,\cdot)$ is a nonnegative function involving the
differential cross-section of the collision dynamics.  Such a high-dimensional
integral form introduces great difficulty to the numerical simulation of the
Boltzmann equation, and people have been using the stochastic method introduced
by Bird \cite{Bird1963, Bird}, known as direct simulation of Monte Carlo
(DSMC), to solve the Boltzmann equation. Due to the fast development of super
computers, in the past decade, a number of deterministic methods have been
proposed to discretize the integral collision term to avoid numerical
oscillations. The most promising method seems to be the Fourier spectral method
\cite{Pareschi1996, Mouhot2006, Hu2017}, including a variety of its variations
such as the conservative version \cite{Gamba2009}, the positivity preserving
version \cite{Pareschi2000}, the steady-state preserving version
\cite{filbet2015steady} and the entropic version \cite{Cai2019}, where the
technique of fast Fourier transform can be applied to accelerate the
computation. These methods has been applied to spatially inhomogeneous problems
in \cite{Wu2013, Wu2014, Dimarco2018}. Other methods include the fast discrete
velocity method \cite{Mouhot2013} and the discontinuous Galerkin method
\cite{Alekseenko2014}.

Another type of spectral method based on global orthogonal polynomials is also
being studied recently \cite{Cai2015,Gamba2018,QuadraticCol}. In this paper, we
follow the work \cite{Cai2015, Gamba2018} and adopt the spectral method based
on Burnett polynomials \cite{Burnett1936}, which has been applied to the
linearized Boltzmann equation \cite{Cai2015, Cai2018}, and shows great
potential to achieve higher numerical efficiency. To focus on the collision
term, we consider only the spatially homogeneous Boltzmann equation, meaning
that the distribution function is uniform in space, and thus we can use a map
$F: \mathbb{R}_+ \rightarrow L^1(\mathbb{R}^3)$ to describe the evolution of
the distribution function:
\begin{equation} \label{eq:Boltzmann}
\begin{aligned}
& \frac{\mathrm{d}F(t)}{\mathrm{d}t} = \mQ[F(t),F(t)],
  \qquad \forall t\in (0,+\infty), \\
& F(0) = f^0 \in L^1(\mathbb{R}^3).
\end{aligned}
\end{equation}
It is well-known that the Boltzmann equation preserves the conservation of
mass, momentum and energy:
\begin{equation} \label{eq:conservation}
\int_{\bbR^3} \begin{pmatrix} 1 \\ \bv \\ |\bv|^2 \end{pmatrix} \mQ[f,f](\bv)
  \dd \bv = 0, \qquad \forall f \in L^1(\mathbb{R}^3).
\end{equation}
Thus we can choose appropriate nondimensionalization such that the initial
value $f^0$ in \eqref{eq:Boltzmann} belongs to the following set:
\begin{displaymath}
\mathcal{S} = \left\{ f \in L^1(\bbR^3) : \int_{\bbR^3}
  \begin{pmatrix} 1 \\ \bv \\ |\bv|^2 \end{pmatrix} f(\bv) \dd \bv =
  \begin{pmatrix} 1 \\ 0 \\ 3 \end{pmatrix} \right\},
\end{displaymath}
and by \eqref{eq:conservation}, for all $t > 0$, we always have $F(t) \in
\mathcal{S}$. According to Boltzmann's H-theorem, the steady-state solution of
\eqref{eq:Boltzmann} is the Maxwellian
\begin{equation}
\label{eq:Maxwellian}
\mathcal{M}(\bv) = \frac{1}{(\sqrt{2\pi})^3}
\exp \left( -\frac{|\bv|^2}{2} \right).
\end{equation}
The Burnett polynomials, which will be used in our discretization, are
orthogonal polynomials associated with the weight function $\mathcal{M}(\bv)$.
Therefore our numerical method can represent this steady-state solution
exactly.

In \cite{QuadraticCol, Hu2018}, a similar method using Hermite polynomials,
which are also orthogonal polynomials associated with the weight function
$\mathcal{M}(\bv)$, is studied. In principle, the spectral methods using
Burnett and Hermite polynomials are essentially equivalent, especially when the
series is truncated up to the same degree. The Hermite spectral method was
introduced long ago by Grad \cite{Grad1949} as the moment method. As mentioned
in \cite[pp. 283]{Grad1958}, the Hermite spectral method is frequently
advantageous due to ``the symmetries inherent in the invariant Cartesian
tensor''. Such an advantage has been utilized in \cite{QuadraticCol}, where the
explicit expressions of all the coefficients in the discretization are
formulated using these symmetries. However, the superiority of the Burnett
polynomials introduced in \cite{Burnett1936} is the fact that they are
eigenfunctions of the linearized collision integral for Maxwell molecules
\cite{Chang1952}. Even for non-Maxwell molecules, as will be shown in this
paper, the coefficients also possess some sparsity due to the rotational
invariance of the collision operator. This will result in a considerably faster
algorithm in the computation, which makes the spectral method with Burnett
polynomials preferable in the simulation.

The same basis functions have been used in \cite{Gamba2018}, where the authors
employed numerical integration to find all the coefficients involved in the
discretization of \eqref{eq:Boltzmann}, but the sparsity in the coefficients
was not utilized in the computation. In this paper, we are going to focus on
the detailed implementation of the algorithm, including a much more accurate
way to compute the coefficients, a detailed analysis of the computational cost,
and the design of the data structure to achieve high computational efficiency.
Meanwhile, we also emphasize the modelling technique introduced in
\cite{QuadraticCol} which allows flexible balancing between computational cost
and modelling error.

The rest of this paper is organized as follows. In Section \ref{sec:general},
we present the framework of the Burnett spectral method to solve the
homogeneous Boltzmann equation. In Section \ref{sec:impl}, the detailed
implementation of the algorithm is introduced. We first give an efficient
method to compute the coefficients in the Burnett spectral expansion, and then
discuss the design of the data structure and the implementation of the
algorithm in detail. Some numerical experiments verifying the efficiency of
the Burnett spectral method are carried out in Section \ref{sec:numerical}.
In Section \ref{sec:proof}, we list the proof of the theorems in Section
\ref{sec:general}. Some concluding remarks are made in Section
\ref{sec:conclusion}.

\section{Framework of the Burnett spectral method}
\label{sec:general}
Burnett polynomials are introduced in \cite{Burnett1936} to study
high-order approximation to the distribution function for a slightly
non-uniform gas. Here we adopt a normalized form and write Burnett
polynomials as
\begin{displaymath}
p_{lmn}(\bv) = \sqrt{\frac{2^{1-l} \pi^{3/2} n!}{\Gamma(n+l+3/2)}}
  L_n^{(l+1/2)} \left( \frac{|\bv|^2}{2} \right) |\bv|^l
  Y_l^m \left( \frac{\bv}{|\bv|} \right), \qquad
l,n \in \bbN, \quad m = -l,\cdots,l,
\end{displaymath}
where we have used Laguerre polynomials
\begin{displaymath}
L_n^{(\alpha)}(x) = \frac{x^{-\alpha} \exp(x)}{n!}
  \frac{\mathrm{d}^n}{\mathrm{d}x^n}
  \left[ x^{n+\alpha} \exp(-x) \right],
\end{displaymath}
and spherical harmonics
\begin{displaymath}
Y_l^m(\bn) = \sqrt{\frac{2l+1}{4\pi} \frac{(l-m)!}{(l+m)!}}
  P_l^m(\cos \theta) \exp(\mathrm{i} m \phi), \qquad
\bn = (\sin \theta \cos \phi, \sin \theta \sin \phi, \cos \theta)^T
\end{displaymath}
with $P_l^m$ being the associate Legendre polynomial:
\begin{displaymath}
P_l^m(x) = \frac{(-1)^m}{2^l l!} (1-x^2)^{m/2}
  \frac{\mathrm{d}^{l+m}}{\mathrm{d}x^{l+m}} \left[ (x^2-1)^l \right].
\end{displaymath}
By the orthogonality of Laguerre polynomials and spherical harmonics,
one can find that
\begin{displaymath}
\int_{\bbR^3} \overline{p_{l_1 m_1 n_1}(\bv)} p_{l_2 m_2 n_2}(\bv)
  \mathcal{M}(\bv) \dd\bv =
\delta_{l_1 l_2} \delta_{m_1 m_2} \delta_{n_1 n_2}.
\end{displaymath}

Now we introduce the basis function $\varphi_{lmn}(\bv)$ as
\begin{equation}
    \varphi_{lmn}(\bv) = p_{lmn}(\bv)\mathcal{M}(\bv).
\end{equation}
For a given distribution function $f \in \mathcal{S}$, we assume that
it has the expansion
\begin{equation} \label{eq:exp_f}
f(\bv) = \sum_{lmn} \tilde{f}_{lmn} \varphi_{lmn}(\bv),
\end{equation}
where the sum is interpreted as
\begin{displaymath}
\sum_{lmn} = \sum_{l=0}^{+\infty} \sum_{m=-l}^l \sum_{n=0}^{+\infty}.
\end{displaymath}
Suppose the corresponding collision term $\mQ[f,f]$ also has the
expansion
\begin{displaymath}
\mQ[f,f](\bv) = \sum_{lmn} \tilde{Q}_{lmn} \varphi_{lmn}(\bv).
\end{displaymath}
By the orthogonality of the Burnett polynomials and the bilinearity of
the operator $\mQ[\cdot,\cdot]$, one can find that
\begin{displaymath}
\tilde{Q}_{lmn} = \sum_{l_1 m_1 n_1} \sum_{l_2 m_2 n_2}
  A_{lmn}^{l_1 m_1 n_2, l_2 m_2 n_2}
  \tilde{f}_{l_1 m_1 n_1} \tilde{f}_{l_2 m_2 n_2},
\end{displaymath}
where
\begin{equation} \label{eq:A}
A_{lmn}^{l_1 m_1 n_2, l_2 m_2 n_2} = \int_{\bbR^3}
  \overline{p_{lmn}(\bv)}
  \mQ[\varphi_{l_1 m_1 n_1}, \varphi_{l_2 m_2 n_2}](\bv) \dd \bv.
\end{equation}
Based on this expansion, it is obvious that \eqref{eq:Boltzmann} is
equivalent to the following ODE system:
\begin{equation}
\label{eq:Boltzmann_ODE} 
\begin{aligned}
& \frac{\mathrm{d} \tilde{F}_{lmn}(t)}{\mathrm{d}t} =
  \sum_{l_1 m_1 n_1} \sum_{l_2 m_2 n_2}
    A_{lmn}^{l_1 m_1 n_2, l_2 m_2 n_2}
    \tilde{F}_{l_1 m_1 n_1}(t) \tilde{F}_{l_2 m_2 n_2}(t), \\
    & \tilde{F}_{lmn}(0) = \tilde{f}_{lmn}^0 :=
  \int_{\bbR^3} \overline{p_{lmn}(\bv)} f^0(\bv) \dd \bv,
\end{aligned}
\end{equation}
where $\tilde{F}_{lmn}(t)$ are the coefficients in the Burnett
series expansion of $F(t)$.

To develop the spectral method, one needs to truncate the Burnett
series to restrict the computation to a finite number of coefficients.
A common choice is to choose a positive integer $M$ and require that
the degree of the polynomial, $l + 2n$, to be less than or equal to
$M$. Thus the spectral method for the homogeneous Boltzmann equation
\eqref{eq:Boltzmann} is
\begin{subequations} \label{eq:ODE}
\begin{align}
\label{eq:ODE1}
& \frac{\mathrm{d} \tilde{F}_{lmn}(t)}{\mathrm{d}t} =
  \sum_{\substack{l_1 m_1 n_1\\[2pt] l_1 + 2n_1 \leqslant M}}
  \sum_{\substack{l_2 m_2 n_2\\[2pt] l_2 + 2n_2 \leqslant M}}
    A_{lmn}^{l_1 m_1 n_2, l_2 m_2 n_2}
    \tilde{F}_{l_1 m_1 n_1}(t) \tilde{F}_{l_2 m_2 n_2}(t), \\
& \tilde{F}_{lmn}(0) = \tilde{f}_{lmn}^0, \qquad l + 2n \leqslant M.
\end{align}
\end{subequations}
These ordinary differential equations can be solved by Runge-Kutta
methods. Naively, the computational cost appears to be
$O(N^3)=O(M^9)$, where $N=(M+1)(M+2)(M+3)/6$ is the total number of
$\tilde{F}_{lmn}$, $l+2n\leq M$. It is worth pointing out that there
is a prefactor $1 / 6^3$ of $O(M^9)$ when we count the number of
coefficients $A_{lmn}^{l_1 m_1 n_2, l_2 m_2 n_2}$, and this prefactor
will be directly brought into the computational cost of the collision
term.

However, the actual computational cost can be reduced to $O(M^8)$ due
to the following sparsity of the coefficients $A_{lmn}^{l_1 m_1 n_2,
l_2 m_2 n_2}$:
\begin{theorem} \label{thm:sparsity}
The coefficient $A_{lmn}^{l_1 m_1 n_2, l_2 m_2 n_2}$ is zero if $m
\neq m_1 + m_2$.
\end{theorem}
By taking into account such sparsity, we can find that the number of
nonzero coefficients $A_{lmn}^{l_1 m_1 n_2, l_2 m_2 n_2}$ is $O(M^8)$
with a prefactor $1/297$. This indicates that the evaluation of the
collision can be efficient for not too large $M$. Interestingly, for
some special collision kernel $B(\cdot,\cdot)$, the computational cost
can be further reduced due to the following extra sparsity of the
coefficients $A_{lmn}^{l_1 m_1 n_2, l_2 m_2 n_2}$:
\begin{theorem} \label{thm:sparsityMaxwell}
If the kernel $B(g,\chi)=\sigma(\chi)$ is independent of $g$, the
coefficient $A_{lmn}^{l_1 m_1 n_2, l_2 m_2 n_2}$ is zero if
$l_1+2n_1+l_2+2n_2 \neq l+2n$.
\end{theorem}
A well-known type of collision kernel satisfying the above condition
is the Maxwell molecules, for which the force between a pair of
molecules is always repulsive and proportional to the fifth power of
their distance. The sparsity stated in the above theorem allows us to
reduce the number of nonzero coefficients to $O(M^7)$. By a numerical
test, we find that the prefactor of $O(M^7)$ is around $1 / (2.2\times
10^{3})$. Moreover, the following result also helps save computational
resources:
\begin{theorem} \label{thm:real}
All coefficients $A_{lmn}^{l_1 m_1 n_1, l_2 m_2 n_2}$ are real.
\end{theorem}

The proof of all the above theorems will be provided in Section
\ref{sec:proof}. They help us save both memory and computational time.
For general collision kernel, by Theorem \ref{thm:sparsity}, we do not
need to store the zero coefficients, and the constraint $m = m_1 +
m_2$ reduces the order of time complexity by $1$. For Maxwell
molecules, Theorem \ref{thm:sparsity} and Theorem
\ref{thm:sparsityMaxwell} reduce the order of time complexity by $2$.
Theorem \ref{thm:real} does not reduce the order, but by realizing
that all the coefficients are real, one can reduce the storage
requirement by a half, and the algorithm can also be made faster by
avoiding some operations between complex numbers. Actually, we can
further reduce the computational cost by using the fact that the
distribution functions are real: since
\begin{displaymath}
\varphi_{lmn}(\bv) = (-1)^m \overline{\varphi_{lmn}(\bv)},
\end{displaymath}
the coefficients in \eqref{eq:exp_f} must satisfy $\tilde{f}_{lmn} =
(-1)^m \overline{\tilde{f}_{l,-m,n}}$ to ensure that $f(\bv)$ is real;
therefore, when solving the ordinary differential equations
\eqref{eq:ODE}, we only need to take into account the case $m
\geqslant 0$, which cuts down the computational cost by a half. The
small prefactor of the computational complexity also indicates the low
computational cost is acceptable if $M$ is not too large.

However, the time complexity $O(M^8)$ (or $O(M^7)$ for Maxwell
molecules) still gives huge computational cost when $M$ is large,
especially when solving spatially
inhomogeneous problems. To make the computation even cheaper, we adopt
the idea in \cite{Cai2015, QuadraticCol} which replaces the right-hand
side of \eqref{eq:ODE1} by $\tilde{Q}_{lmn}^{*}(t)$, defined as
\begin{equation} \label{eq:Q_star}
\tilde{Q}_{lmn}^{*}(t) = \left\{
  \begin{array}{ll}
  \displaystyle \sum \limits_{\substack{l_1 m_1 n_1\\[2pt] l_1 + 2n_1 \leqslant M_0}}
  \sum\limits_{\substack{l_2 m_2 n_2\\[2pt] l_2 + 2n_2 \leqslant M_0}}
    A_{lmn}^{l_1 m_1 n_2, l_2 m_2 n_2}
    \tilde{F}_{l_1 m_1 n_1}(t) \tilde{F}_{l_2 m_2 n_2}(t),  
  & \text{if } l+2n \leqslant M_0,\\[28pt]
    -\mu_{M_0}\tilde{F}_{lmn}(t), & \text{otherwise}.
  \end{array}
\right. 
\end{equation}
In practice, one can set $M_0$ to be much less than $M$. Thus the
quadratic form is only applied to the first few coefficients whose
associate polynomials have degree less than or equal to $M_0$. When
$l+2n > M_0$, similar to the BGK-type models, we let the coefficient
decay to zero exponentially at a constant rate $\mu_{M_0}$. Thereby we
get the new model
\begin{equation} \label{eq:ODE_final}
\frac{\mathrm{d} \tilde{F}_{lmn}(t)}{\mathrm{d}t} =
  \tilde{Q}_{lmn}^{*}(t).
\end{equation}
As in \cite{QuadraticCol, Cai2015}, we choose the decay rate
$\mu_{M_0}$ to be the spectral radius of the linearized collision
operator $\mathcal{L}_{M_0}: \mathcal{S}_{M_0}\rightarrow
\mathcal{S}_{M_0}$ defined as
\begin{equation}
  \label{eq:linear_Boltzmann}
  \mathcal{L}_{M_0}[f](\bv) = 
    \sum_{\substack{l_1 m_1 n_1\\[2pt] l_1 + 2n_1 \leqslant M_0}}
    \sum_{\substack{n_2 \leqslant (M_0 - l_1)/2}}
    \left(A_{l0n_1}^{l0n_2, 000} + A_{l0n_1}^{000, l0n_2} \right)
      \tilde{f}_{l_1 m_1 n_2}\psi_{l_1m_1n_1}(\bv),  
\end{equation}
where $\mathcal{S}_{M_0} = \text{span}\{\psi_{lmn}(\bv) : l + 2n
\leqslant M_0 \} \cap \mathcal{S}$. We refer the readers to
\cite{Cai2015} for more details.

By now, we have obtained the ordinary differential equations to
approximate the homogeneous Boltzmann equation \eqref{eq:Boltzmann}
under the framework of Burnett polynomials. In order to complete this
algorithm, we still need to find the values of the coefficients
$A_{lmn}^{l_1m_1n_1,l_2m_2n_2}$, which will be detailed in the
following section. Moreover, the implementation of the algorithm will
be discussed deeply to obtain optimal efficiency.


\section{Implementation of the algorithm} 
\label{sec:impl} 
To implement the algorithm, we first need to find the values of the
coefficients $A_{lmn}^{l_1 m_1 n_1, l_2 m_2 n_2}$. A formula of these
coefficients has been given in \cite{Kumar}, which reads
\begin{equation}\label{eq:A_lmn}
A_{lmn}^{l_1 m_1 n_1, l_2 m_2 n_2} =
  \sum_{l_3 m_3 n_3} \sum_{l_4 m_4 n_4} \sum_{n_4'}
  \overline{\left( \begin{array}{c|c}
    n_3 l_3 m_3 & nlm \\ n_4' l_4 m_4 & 000
  \end{array} \right)} \left( \begin{array}{c|c}
    n_3 l_3 m_3 & n_1 l_1 m_1 \\ n_4 l_4 m_4 & n_2 l_2 m_2
  \end{array} \right) V_{n_4 n_4'}^l,
\end{equation}
where
\begin{equation}\label{eq:def_Vnn}
\begin{split}
V_{n n'}^l &= \frac{1}{8\sqrt{2} \pi^{5/2}}
  \sqrt{\frac{n!n'!}{\Gamma(n+l+3/2) \Gamma(n'+l+3/2)}} 
  \int_0^{+\infty} \int_0^{\pi} B(g,\chi)
  \left( \frac{g^2}{4} \right)^{l+1} \times {} \\
& \qquad L_n^{(l+1/2)} \left( \frac{g^2}{4} \right)
  L_{n'}^{(l+1/2)} \left( \frac{g^2}{4} \right)
  [(2l+1)^2 P_l(\cos \chi) - 1]
  \exp \left( -\frac{g^2}{4} \right) \dd\chi \dd g,
\end{split}
\end{equation}
and the notation $(\cdot \mid \cdot)$ denotes Talmi coefficients for
equal mass molecules \cite{Talmi1952}. Due to the sparsity of the
Talmi coefficients, the computational cost for evaluating all
$A_{lmn}^{l_1 m_1 n_1, l_2 m_2 n_2}$ required in \eqref{eq:ODE} is
$O(M_0^{14})$. Thus, computing the Talmi coefficients is not an
easy task. The reference \cite{Bakri1967} provides a possible
implementation, but the formula involves Wigner 3-$j$ and 9-$j$
symbols, which are also difficult to obtain. Below we are going to
propose another method to compute these coefficients
$A_{lmn}^{l_1 m_1 n_1, l_2 m_2 n_2}$ based on the work
\cite{QuadraticCol}. The method also has computational cost
$O(M_0^{14})$, but is much easier to implement.

\subsection{Computation of coefficients
  $A_{lmn}^{l_1 m_1 n_1, l_2 m_2 n_2}$}
In \cite{QuadraticCol}, we have calculated the expansion coefficients
of the quadratic collision term $\mQ[f,f](\bv)$ under the framework of
Hermite spectral method. Since both Hermite polynomials and Burnett
polynomials are orthogonal polynomials associated with the same weight
function, we can express the Burnett polynomials by linear
combinations of the Hermite polynomials, and then the coefficients
$A_{lmn}^{l_1 m_1 n_1, l_2 m_2 n_2}$ naturally become a linear
combination of the corresponding coefficients in the Hermite spectral
method. Since the expressions for the coefficients in the Hermite
spectral method have been worked out explicitly in
\cite{QuadraticCol}, we do not need to bother using the complicated
symbols in the quantum theory to find the values of $A_{lmn}^{l_1 m_1
n_1, l_2 m_2 n_2}$.

Mathematically, the above framework can be formulated as below. In
\cite{QuadraticCol}, Hermite polynomial $H^{k_1 k_2 k_3}(\bv)$ is
defined as
\begin{equation}
  \label{eq:basis}
  H^{k_1 k_2 k_3}(\bv) = \frac{(-1)^n}{\mathcal{M}(\bv)}
  \frac{\partial^{k_1+k_2+k_3}}{\partial v_1^{k_1} \partial
    v_2^{k_2} \partial v_3^{k_2}} \mathcal{M}(\bv), \quad \forall k_1, k_2, k_3 \in \mathbb{N},
\end{equation}
where $\mathcal{M}(\bv)$ is given in \eqref{eq:Maxwellian}. We would
like to express the Burnett polynomials as
\begin{equation}
  \label{eq:exp_he_bur}
  p_{lmn}(\bv) =
  \sum_{(k_1,k_2,k_3) \in I_{l + 2n}}\frac{1}{\factorialk}C_{lmn}^{k_1k_2k_3}H^{\indexk}(\bv), 
\end{equation}
where $I_{l+2n}$ is the index set
\begin{equation}
  \label{eq:indexset}
  I_{l+2n} = \{(k_1, k_2, k_3) \in \bbN^3 \mid k_1+k_2+k_3 = l+2n\},
\end{equation}
and the coefficients $C_{lmn}^{k_1k_2k_3}$ can be calculated as 
\begin{equation} \label{eq:C}
C_{lmn}^{\indexk} =
  \int_{\bbR^3}p_{lmn}(\bv)H^{\indexk}(\bv)\mM(\bv) \dd \bv
\end{equation}
based on the orthogonality of Hermite polynomials
\begin{equation}
  \label{eq:Her_orth}
  \int_{\bbR^3}H^{k_1k_2k_3}(\bv)H^{\indexl}(\bv)\mM(\bv)  \dd \bv
  = \delta_{k_1l_1}\delta_{k_2l_2}\delta_{k_3l_3}k_1!k_2!k_3!.
\end{equation}
Note that when $l + 2n \neq k_1 + k_2 + k_3$, i.e. the degrees of
$H^{k_1 k_2 k_3}$ and $p_{lmn}$ are not equal, the coefficient
$C_{lmn}^{k_1 k_2 k_3}$ defined by \eqref{eq:C} is zero due to the
orthogonality of both polynomials. Once $C_{lmn}^{k_1k_2k_3}$ is
obtained, we just need to substitute \eqref{eq:exp_he_bur} to the
definition of $A_{lmn}^{l_1 m_1 n_1, l_2 m_2 n_2}$ \eqref{eq:A}, which
results in the following formula for these coefficients:
\begin{equation}
  \label{eq:A_detail}
  \begin{aligned}
    A_{lmn}^{l_1m_1n_1, l_2m_2n_2} 
    &= \sum_{k \in I_{l + 2n}}\sum_{i \in I_{l_1 +2n_1}} \sum_{j \in I_{l_2 + 2n_2}}
      \overline{C_{lmn}^{k_1k_2k_3}} C_{l_1m_1n_1}^{\indexi}
        C_{l_2m_2n_2}^{\indexj} \times{} \\
    & \qquad \frac{1}{i_1! i_2! i_3! j_1! j_2! j_3! k_1! k_2! k_3!}
      \int_{\bbR^3}H^{\indexk}(\bv)
        \mQ\left[H^{\indexi}\mM,H^{\indexj}\mM\right](\bv)\dd \bv.
  \end{aligned}
\end{equation}
The second line of \eqref{eq:A_detail} has already been evaluated in
\cite[Theorem 1 \& 2]{QuadraticCol} (denoted as
$A_{\indexk}^{\indexi,\indexj}$ therein). Thus we will focus only on
the computation of the coefficients $C_{lmn}^{\indexk}$ below.

Define 
\begin{equation}
  \label{eq:S}
  S_{-1} = \frac{1}{2}(v_1  - \imag v_2), \quad S_0 = v_3, \quad S_1 =
  -\frac{1}{2}(v_1 + \imag v_2), 
\end{equation}
and
\begin{equation}
  \label{eq:gamma}
  \gamma_{lm}^{\mu} = \sqrt{\frac{[l + (2\delta_{1, \mu}-1) m + \delta_{1,
      \mu}][l - (2\delta_{-1, \mu} - 1)m + \delta_{-1, \mu}]}{(2l
    -1)(2l+1)}},
\end{equation}
the recursive formula of the basis functions \cite{Cai2018} is
\begin{equation}
  \label{eq:recursive}
  \begin{aligned}
    S_{\mu}\psi_{lmn}(\bv) &= \frac{1}{2^{|\mu|}}
      \left[\sqrt{2(n+l)+3}\gamma_{l+1, m}^{\mu}\psi_{l+1, m+\mu, n}(\bv) -
      \sqrt{2n}\gamma_{l+1, m}^{\mu}\psi_{l+1, m+\mu, n-1}(\bv) \right. \\
    &\left. + (-1)^{\mu}\sqrt{2(n+l)+1} \gamma_{-l,m}^{\mu}\psi_{l-1, m+\mu, n}(\bv)
      -(-1)^{\mu}\sqrt{2(n+1)}\gamma_{-l,m}^{\mu}\psi_{l-1, m+\mu,n+1}(\bv) \right],
  \end{aligned}
\end{equation}
where we set $\psi_{lmn}(\bv) = 0$ if $|m| > l$ or either of $l, n$ is
negative. Equations of $C_{lmn}^{k_1 k_2 k_3}$ can be obtained by
multiplying \eqref{eq:recursive} with $H^{\indexk}(\bv)$ and taking
integration with respect to $\bv$ on both sides. The integral of the
right-hand side can be written straightforwardly as expressions of 
$C_{lmn}^{k_1 k_2 k_3}$, while for the left-hand side, we need to use
the recursion formula of Hermite polynomials
\begin{equation}
  \label{eq:recursive_h}
  v_s H^{\indexk}(\bv) = H^{k_1+\delta_{1s}, k_2+\delta_{2s},k_3+\delta_{3s}}(\bv)
  + k_sH^{k_1-\delta_{1s}, k_2-\delta_{2s},k_3-\delta_{3s}}(\bv), \quad s = 1,2,3.
\end{equation}
Since both the Hermite and Burnett polynomials are orthogonal
polynomials, integrals including the product of polynomials of
different degrees all vanish. Therefore when applying the above
operations, we choose $k_1, k_2, k_3$ such that $k_1 + k_2 + k_3 =
l+2n+1$, and the resulting equations for $\mu = -1,0,1$ are
respectively
\begin{equation}
  \label{eq:recursive_c}
  \begin{aligned}
    a_{l,m+1,n}^{(-1)}C_{l+1, m, n}^{k_1 k_2 k_3} + b_{l,m+1,n}^{(-1)}C_{l-1, m,
        n+1}^{k_1 k_2 k_3} &= \frac{1}{2}k_1C_{l,m+1,n}^{k_1-1,k_2,k_3} -
      \frac{\imag}{2} k_2 C_{l,m+1,n}^{k_1,k_2-1,k_3}, \\
      a_{l,m,n}^{(0)}C_{l+1, m, n}^{k_1 k_2 k_3} + b_{l,m,n}^{(0)}C_{l-1, m,
        n+1}^{k_1 k_2 k_3} &= k_3C_{l,m,n}^{k_1,k_2,k_3-1}, \\
      a_{l,m-1,n}^{(1)}C_{l+1, m, n}^{k_1 k_2 k_3} + b_{l,m-1,n}^{(1)}C_{l-1, m,
        n+1}^{k_1 k_2 k_3} &= -\frac{1}{2}k_1C_{l,m-1,n}^{k_1-1,k_2,k_3} -
      \frac{\imag}{2} k_2C_{l,m-1,n}^{k_1,k_2-1,k_3},
    \end{aligned}     
\end{equation}
where
\begin{equation}
  \label{eq:coe_a}
  a_{lmn}^{(\mu)} = \frac{1}{2^{|\mu|}}\sqrt{(2(n+l)
    +3)}\gamma_{l+1,n}^{\mu},
  \quad b_{lmn}^{(\mu)} = \frac{(-1)^{\mu+1}}{2^{|\mu|}}\sqrt{2(n+1)
  }\gamma_{-l,n}^{\mu}, \quad \mu = -1, 0, 1,
\end{equation}
and we have exchanged the left-hand side and the right-hand side since
we would like to solve the coefficients $C_{l+1, m, n}^{k_1 k_2 k_3}$
and $C_{l-1, m, n+1}^{k_1 k_2 k_3}$ from the above equations. This is
possible since on the left-hand sides of \eqref{eq:coe_a}, the sum of
the superscripts $k_1 + k_2 + k_3$ is always greater than the same sum
on the right-hand sides. Hence we can solve all the coefficients
$C_{lmn}^{k_1 k_2 k_3}$ by the order of $k_1 + k_2 + k_3$, so that the
right-hand sides of \eqref{eq:recursive_c} are always known. To start
the computation, we need the ``initial condition'' $C_{000}^{000}=1$,
and the ``boundary conditions'' $C_{lmn}^{\indexk} = 0$ if $|m| > l$
or either of $l, n$ is negative. It is not difficult to see that the
computational cost for each coefficient is $O(1)$, and thus the time
complexity for computing all the coefficients $C_{lmn}^{\indexk}$ with
$l+2n = k_1 + k_2 + k_3 \leqslant M_0$ is $O(M_0^5)$.

Now we come back to \eqref{eq:A_detail}. From Theorem
\ref{thm:sparsity}, it is known that the total number of nonzero
$A_{lmn}^{l_1m_1n_1, l_2m_2n_2}$ is $O(M_0^8)$, and the computational
cost of each summation symbol on the right-hand side of
\eqref{eq:A_detail} is at most $O(M_0^2)$ from the definition of the
index set $I_{M_0}$. Therefore, based on the knowledge of the second
line of \eqref{eq:A_detail}, the total computational cost for all the
coefficients $A_{lmn}^{l_1m_1n_1, l_2m_2n_2}$ is $O(M_0^{14})$. In
fact, to get the second line of \eqref{eq:A_detail}, the computational
cost is only $O(M_0^{12})$ as stated in \cite{QuadraticCol}. Thus the
overall complexity for finding $A_{lmn}^{l_1 m_1 n_1, l_2 m_2 n_2}$ is
$O(M_0^{14})$. Here we emphasize again that such a computational cost
is only for the precomputation, which needs to be done only once.

Finally, we would like to comment that the computational cost for
$A_{lmn}^{l_1 m_1 n_1, l_2 m_2 n_2}$ can be further reduced to one
eighth by using the symmetry of Burnett polynomials
\begin{equation}
  \label{eq:odevity}
  p_{lmn}(v_1, v_2, -v_3) = (-1)^{l+m}p_{lmn}(v_1, v_2, v_3), 
\end{equation}
which means the coefficient $C_{lmn}^{\indexk}$ is nonzero only if
$(l+m) - k_3$ is even. By now, we have been able to make the whole
algorithm work, and the rest of this section will be devoted to our
detailed implementation of \eqref{eq:Q_star}, including the design of
the data structure and the detailed steps of our fast algorithm.

\subsection{Data structure: storage of the coefficients}
The optimal data structure to store the coefficients $\tilde{f}_{lmn}$
and $A_{lmn}^{l_1 m_1 n_1, l_2 m_2 n_2}$ should require minimum
``jumps'' in the memory, which means the order of data usage should
match the storage of the data as much as possible. In what follows, we
are going to show by illustration how the data are arranged to achieve
optimal continuity.

\subsubsection{Storage of the coefficients $\tilde{f}_{lmn}$}
Suppose we need to store the coefficients $\tilde{f}_{lmn}$ for all
$l+2n \leqslant M$. We store all the coefficients in a one-dimensional
continuous array. This array can be viewed as the concatenation of
$2M+1$ sections, and each section contains all the coefficients for a
given $m$. Inside each section, the coefficients are ordered as shown
in Figure \ref{fig:f}, where the value of $l+2n$ (the degree of the
corresponding Burnett polynomial) is increasing, and when $l+2n$ is a
constant, the value of $l$ is increasing.

\begin{figure}[!ht]
\centering
\subfigure[Storage pattern of $\tilde{f}_{lmn}$]{
\includegraphics[width=.7\textwidth]{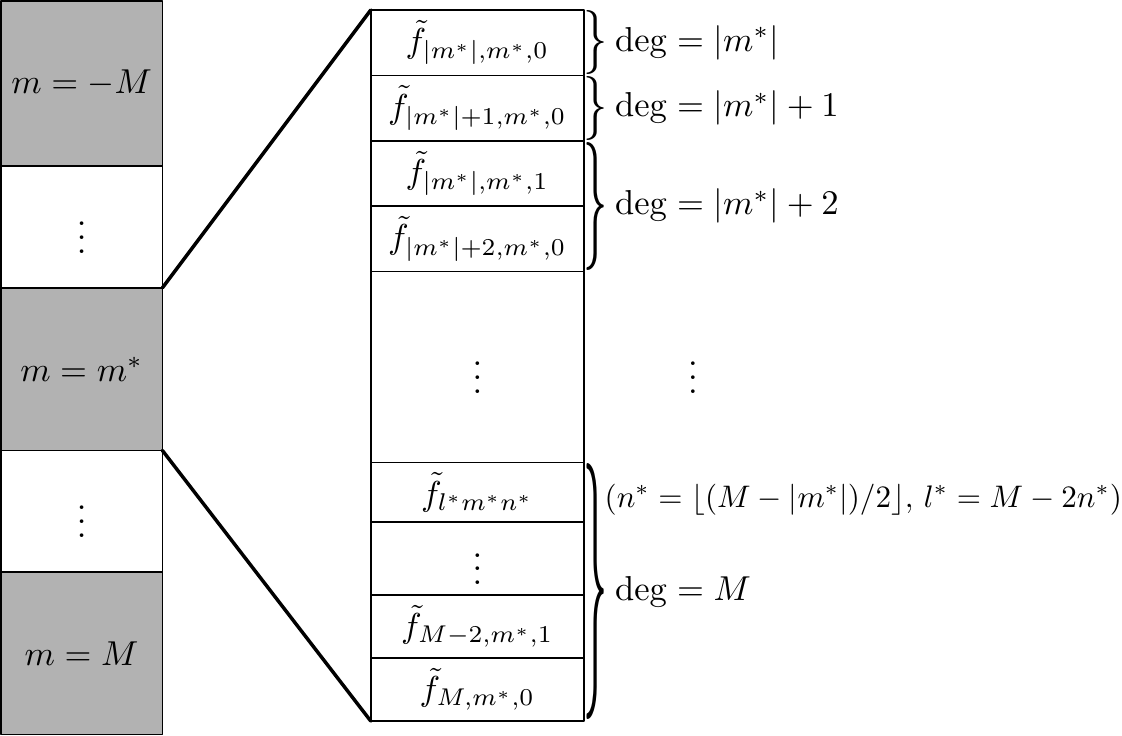}
}\qquad\quad
\subfigure[Example: $M=2$]{
\includegraphics[width=.17\textwidth]{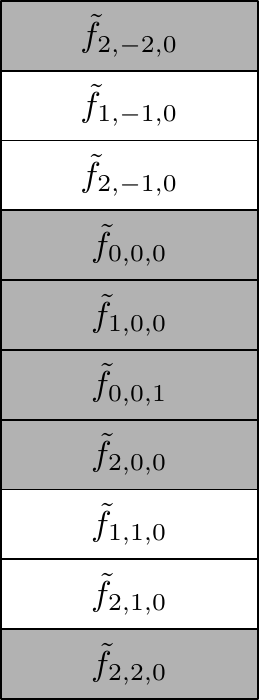}
}
\caption{Storage pattern of $\tilde{f}_{lmn}$ showing the two-level
structure of the array. The left column of (a) shows that the first-level
decomposition of the array (one section for each $m$), and the right
column of (a) shows how the elements are stored in the second-level
structure. (b) presents an example with $M=2$.}
\label{fig:f}
\end{figure}

By this storage scheme, for any given $m$, the coefficients associated
with the polynomials of degree less than or equal to $M_0$ are
continuously stored, which makes it easier to perform the
matrix-vector multiplication in the numerical algorithm and achieve
good cache hit ratio.

\subsubsection{Storage of the coefficients $A_{lmn}^{l_1 m_1 n_1, l_2 m_2
n_2}$}
\begin{figure}[!ht]
\centering
\includegraphics[width=\textwidth]{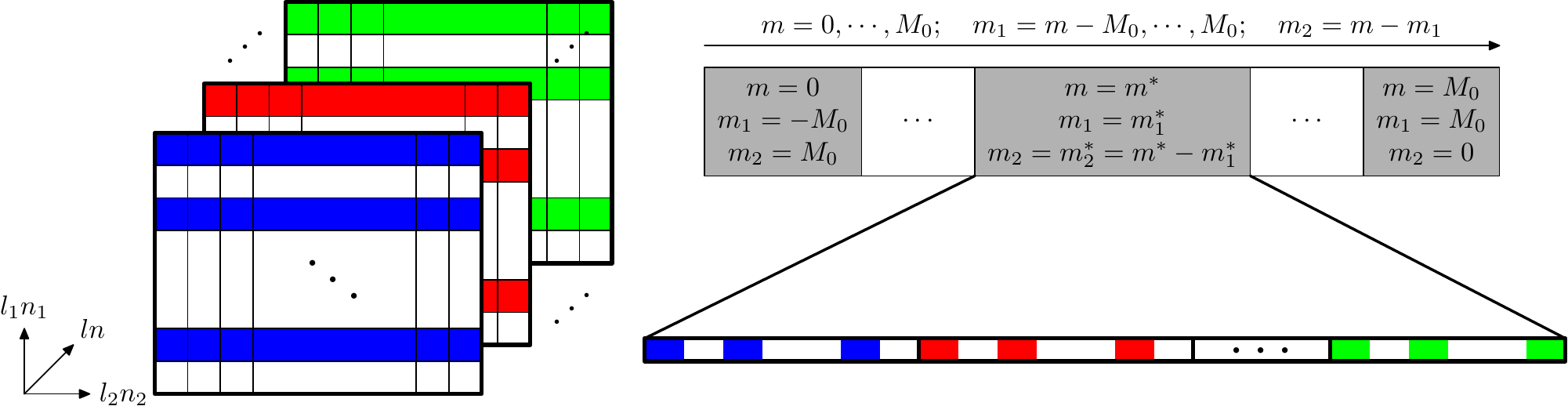}
\caption{Storage pattern of $A_{lmn}^{l_1 m_1 n_1, l_2 m_2 n_2}$. The
left column shows the three-dimensional view of the coefficients for
given $m$, $m_1$ and $m_2$. The right column gives the general
two-level structure of the array.}
\label{fig:A}
\end{figure}
Similar to the storage of $\tilde{f}_{lmn}$, all the coefficients
$A_{lmn}^{l_1 m_1 n_1, l_2 m_2 n_2}$ are also stored in a continuous
array, which can again be considered as the concatenation of a number
of sections. Each section contains all the coefficients for given $m$,
$m_1$ and $m_2$. Noting that the range of $m$ is from $0$ to $M_0$
while the range of $m_1$ and $m_2$ is from $-M_0$ to $M_0$, we can
find that the total number of sections is $(3M_0 + 2)(M_0 + 1)/2$.
Once $m$ is given, the two indices $l$ and $n$ can be viewed as a
one-dimensional index $ln$ by our ordering rule in the storage pattern
of $\tilde{f}_{lmn}$.  Similarly, $l_1 n_1$ and $l_2 n_2$ can also be
regarded as one-dimensional indices.  Thus, once $m$, $m_1$ and $m_2$
are given, the coefficients $A_{lmn}^{l_1 m_1 n_1, l_2 m_2 n_2}$ can
be considered as a three-dimensional array, whose three indices are
$ln$, $l_1 n_1$ and $l_2 n_2$ (see the left column of Figure
\ref{fig:A}). Its storage is a simple flattening the three-dimensional
array and is illustrated in the right column of Figure \ref{fig:A}.


\subsection{Details of the algorithm}
Based on the above data structure, the computation of
$\tilde{Q}_{lmn}^*$ can be implemented very efficiently. The general
procedure is as follows:
\renewcommand{\thealgorithm}{ }
\begin{algorithm}[ht]
\label{alg:Q}
\caption{Algorithm to Calculate $\tilde{Q}_{lmn}^*$}

\begin{algorithmic}[1]
\begin{minipage}{0.8\textwidth}
\For {$m$ from  $0$ to $M$}
\If{$m > M_0$}
\ForAll{$l,n$ satisfying $l \geqslant m$, $l+2n \leqslant M$}
\State $\tilde{Q}_{lmn}^* \gets -\mu_{M_0} \tilde{f}_{lmn}$
\EndFor
\Else
\ForAll{$l,n$ satisfying $l \geqslant m$, $M_0 < l+2n \leqslant M$}
\State $\tilde{Q}_{lmn}^* \gets -\mu_{M_0} \tilde{f}_{lmn}$
\EndFor
\For {$m_1$ from  $m-M_0$  to $M_0$}
\State $m_2 \gets m - m_1$
\ForAll{$l,n$ satisfying $l \geqslant m$, $l+2n \leqslant M_0$}
\State \rule[-.3\baselineskip]{0pt}{1.6\baselineskip}
  $\displaystyle \tilde{Q}_{lmn}^* \gets \sum_{l_1+2n_1 \leqslant
  M_0} \sum_{l_2+2n_2 \leqslant M_0} A_{lmn}^{l_1 m_1 n_1, l_2 m_2 n_2}
  \tilde{f}_{l_1 m_1 n_1} \tilde{f}_{l_2 m_2 n_2}$
\EndFor
\EndFor
\EndIf
\EndFor
\ForAll{$m = -M, \cdots, -1$ and $l,n$ satisfying $l \geqslant |m|$, $l+2n \leqslant M$}
\State $\tilde{Q}_{lmn}^* \gets (-1)^m \overline{\tilde{Q}_{l,-m,n}^*}$
\EndFor
\end{minipage}
\end{algorithmic}
\end{algorithm}

It is worth noting that line 10 can be implemented by two
matrix-vector multiplications:
\begin{align}
\label{eq:tilde_g}
1: & \qquad \tilde{g}_{lmn}^{l_1 m_1 n_1} = \sum_{l_2 + 2n_2
  \leqslant M_0} A_{lmn}^{l_1 m_1 n_1, l_2 m_2 n_2}
    \tilde{f}_{l_2 m_2 n_2}, \\
\label{eq:tilde_Q}
2: & \qquad \tilde{Q}_{lmn}^* = \sum_{l_1 + 2n_1 \leqslant M_0}
  \tilde{g}_{lmn}^{l_1 m_1 n_1} \tilde{f}_{l_1 m_1 n_1}.
\end{align}
Using the storage pattern illustrated in Figure \ref{fig:f} and Figure
\ref{fig:A}, the matrix entries and the vector components involved in
the above operations are automatically continuously stored. The
details are illustrated in Figure \ref{fig:algo}. The left column of
Figure \ref{fig:algo} provides the color coding of the vectors. Each
vertical strip denotes the data structure represented on the right
column of Figure \ref{fig:f}. Since the coefficients with $l + 2n
\leqslant M_0$ and the coefficients with $l + 2n > M_0$ are treated
differently, we distinguish these two parts by shading with slanted
lines. The middle column shows the computation of
$\tilde{g}_{lmn}^{l_1 m_1 n_1}$ (equation \eqref{eq:tilde_g}) for
given $m$ and $m_1$, which is in fact just one matrix-vector
multiplication based on our data structure. The color coding of the
matrix $A_{lmn}^{l_1 m_1 n_1, l_2 m_2 n_2}$ is the same as Figure
\ref{fig:A}. The right column gives the computation of
$\tilde{Q}_{lmn}^*$, which contains a matrix-vector multiplication (for
degree less than or equal to $M_0$, equation \eqref{eq:tilde_Q}) and a
vector scaling (for degree greater than $M_0$). The matrix
$\tilde{g}_{lmn}^{l_1 m_1 n_1}$ is a reshaping of the vector in the
middle column. By comparing the matrix form and the vector form of $A$
and $\tilde{g}$, one can observe our data structure automatically
corresponds to the row-major order of these matrices, which makes it
easy to use optimized BLAS libraries such as ATLAS \cite{Whaley2001}
to achieve high numerical efficiency.

\begin{figure}[!ht]
\centering
\includegraphics[width=\textwidth]{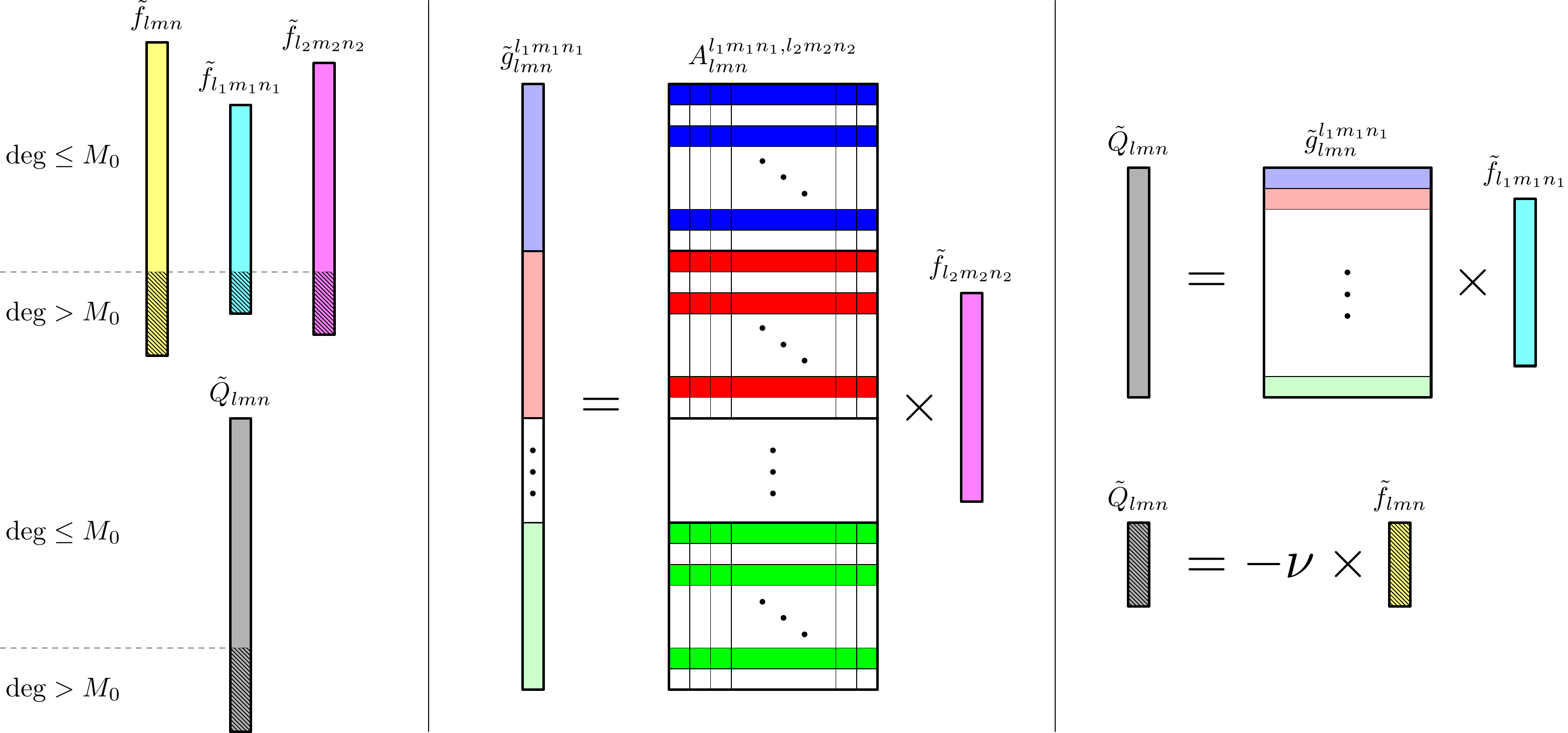}
\caption{Illustration of the algorithm for given $m$, $m_1$ and $m_2$.}
\label{fig:algo}
\end{figure}


\section{Numerical examples} \label{sec:numerical}
In this section, we will show some results of our numerical simulation. In all
the numerical experiments, we consider the inverse-power-law model, for which
the repulsive force between two molecules is proportional to $r^{-\eta}$, with
$r$ and $\eta$ being, respectively, the distance between the two molecules and
a given positive constant. The details about this model can be found in
\cite{Bird}. In all the tests, we use the classical fourth-order Runge-Kutta
method to the equations \eqref{eq:ODE} numerically for some given $M_0$ and
$M$, and the time step is chosen as $\Delta t = 0.01$.

For visualization purposes, we define integration operators $\mathcal{I}_1:
L^1(\mathbb{R}^3) \rightarrow L^1(\mathbb{R})$ and  $\mathcal{I}_2:
L^1(\mathbb{R}^3) \rightarrow L^1(\mathbb{R}^2)$ by
\begin{displaymath}
(\mathcal{I}_1 f)(v_1) = \int_{\mathbb{R}} \int_{\mathbb{R}}
  f(\bv) \,\mathrm{d}v_2 \,\mathrm{d}v_3, \qquad \forall f\in L^1(\mathbb{R}^3),
\end{displaymath}
and
\begin{displaymath}
(\mathcal{I}_2 f)(v_1, v_2) = \int_{\mathbb{R}} f(\bv) \,\mathrm{d}v_3,
  \qquad \forall f \in L^1(\mathbb{R}^3).
\end{displaymath}
These 1D and 2D functions are actually marginal distribution functions (MDFs).
We will only show the plots for these MDFs due to the difficulty in plotting
three-dimensional functions. 

Besides, we are also interested in the evolution of the moments of the
distribution function, especially the heat flux $q_i$ and the stress tensor
$\sigma_{ij}$. For a given distribution function $f \in \mathcal{S}$, they are
defined as
\begin{equation}
  \label{eq:heatflux_stress}
  q_i = \frac{1}{2}\int_{R^3}|\bv|^2 v_i f(\bv) \dd \bv, \qquad
  \sigma_{ij} = \int_{\bbR^3}\left(
    v_i v_j - \frac{1}{3}\delta_{ij}|\bv|^2
  \right) f(\bv) \dd \bv, \qquad i, j = 1, 2, 3. 
\end{equation}
The relations between these moments and the coefficients are
\begin{equation}
  \label{eq:moment_q_sigma}
  \begin{aligned}
& q_1 = \sqrt{5}{\rm Re}(\tilde{f}_{111}), \quad
  q_2 = -\sqrt{5}{\rm Im}(\tilde{f}_{111}), \quad
  q_3 = -\sqrt{5/2} \tilde{f}_{101}, \\
& \sigma_{11} = \sqrt{2}{\rm Re}(\tilde{f}_{220}) - \tilde{f}_{200}/\sqrt{3}, \quad
  \sigma_{12} = -\sqrt{2}{\rm Im}(\tilde{f}_{220}), \quad
  \sigma_{13} = -\sqrt{2}{\rm Re}(\tilde{f}_{210}), \\
& \sigma_{22} = -\sqrt{2}{\rm Re}(\tilde{f}_{220}) - \tilde{f}_{200}/\sqrt{3}, \quad
  \sigma_{23} = \sqrt{2}{\rm Im}(\tilde{f}_{210}), \quad
  \sigma_{33} = 2\tilde{f}_{200}/\sqrt{3}.
  \end{aligned}
\end{equation}

\subsection{BKW (Bobylev-Krook-Wu) solution}
In this example, we study the Maxwell gas whose the power index $\eta$
equals $5$. In this case, the kernel $B(g,\chi)$ turns out to be
independent of $g$ (therefore denoted by $B(\chi)$ below), and it is
given in \cite{Bobylev1984, krook1977exact} that the spatially
homogeneous Boltzmann equation \eqref{eq:Boltzmann} admits an exact
solution $F(t) = f^{[\tau(t)]}$, where
\begin{align*}
& \tau(t) = 1 - \frac{2}{5}\exp(-\lambda t),
  \qquad \lambda = \frac{\pi}{2} \int_0^{\pi} B(\chi) \sin^2 \chi \dd \chi, \\
& f^{[\tau]}(\bv) = (2\pi \tau)^{-3/2} \exp \left( -\frac{|\bv|^2}{2\tau} \right)
  \left[ 1 + \frac{1-\tau}{\tau}
    \left( \frac{|\bv|^2}{2\tau} - \frac{3}{2} \right) \right].
\end{align*}
The initial MDFs are plotted in Figure \ref{fig:ex1_init}, in which the contour
lines for exact functions and their numerical approximation are hardly
distinguishable, with the number $M = 20$. 

\begin{figure}[!ht]
\centering
\subfigure[Initial MDF $\mathcal{I}_1 f^0$\label{fig:ex1_init_1d}]{%
  \includegraphics[width=0.28\textwidth, clip]{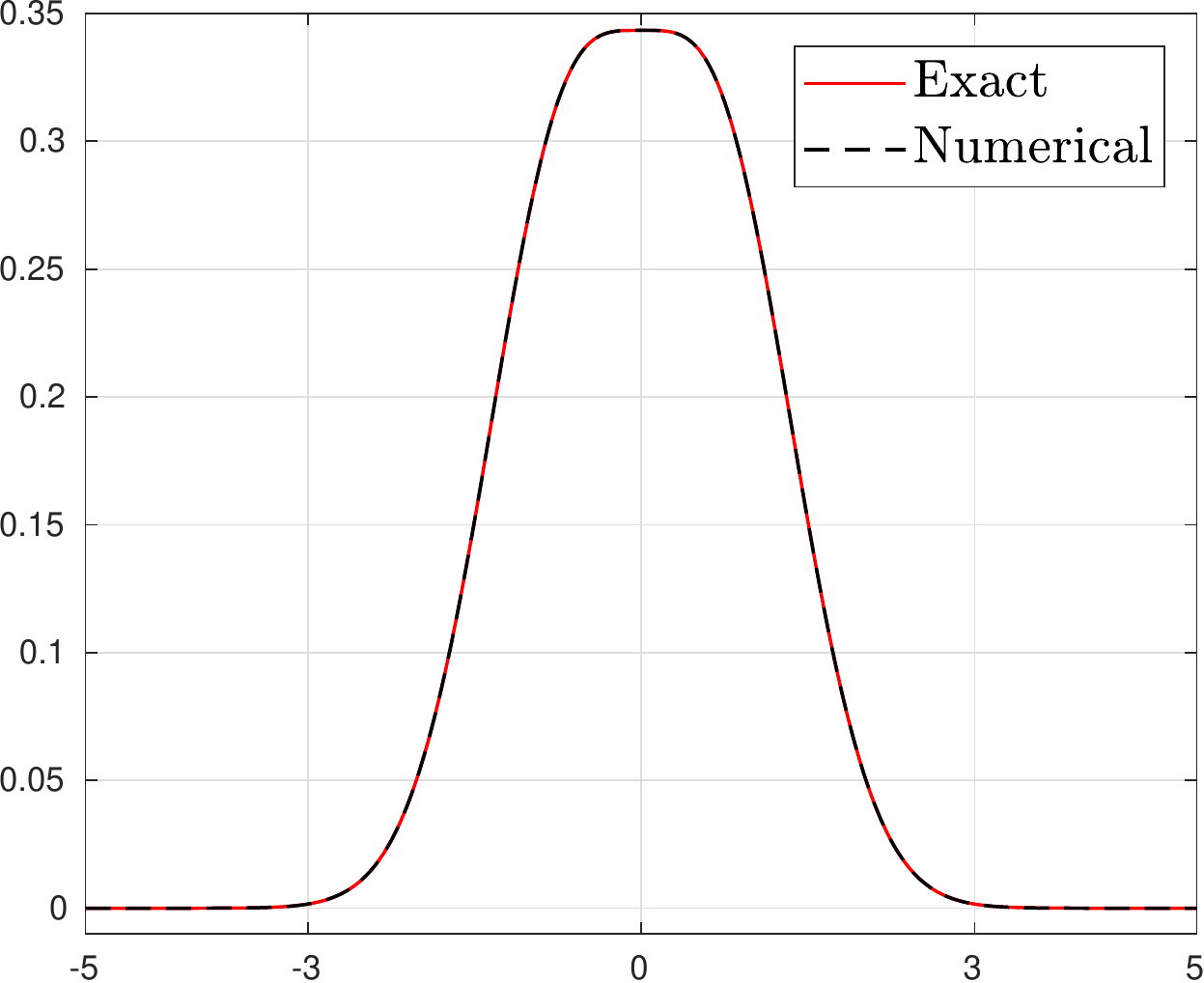}    
}\quad
\subfigure[Contours of $\mathcal{I}_2 f^0$\label{fig:ex1_init_2d_contour}]{%
  \includegraphics[width=0.28\textwidth, clip]{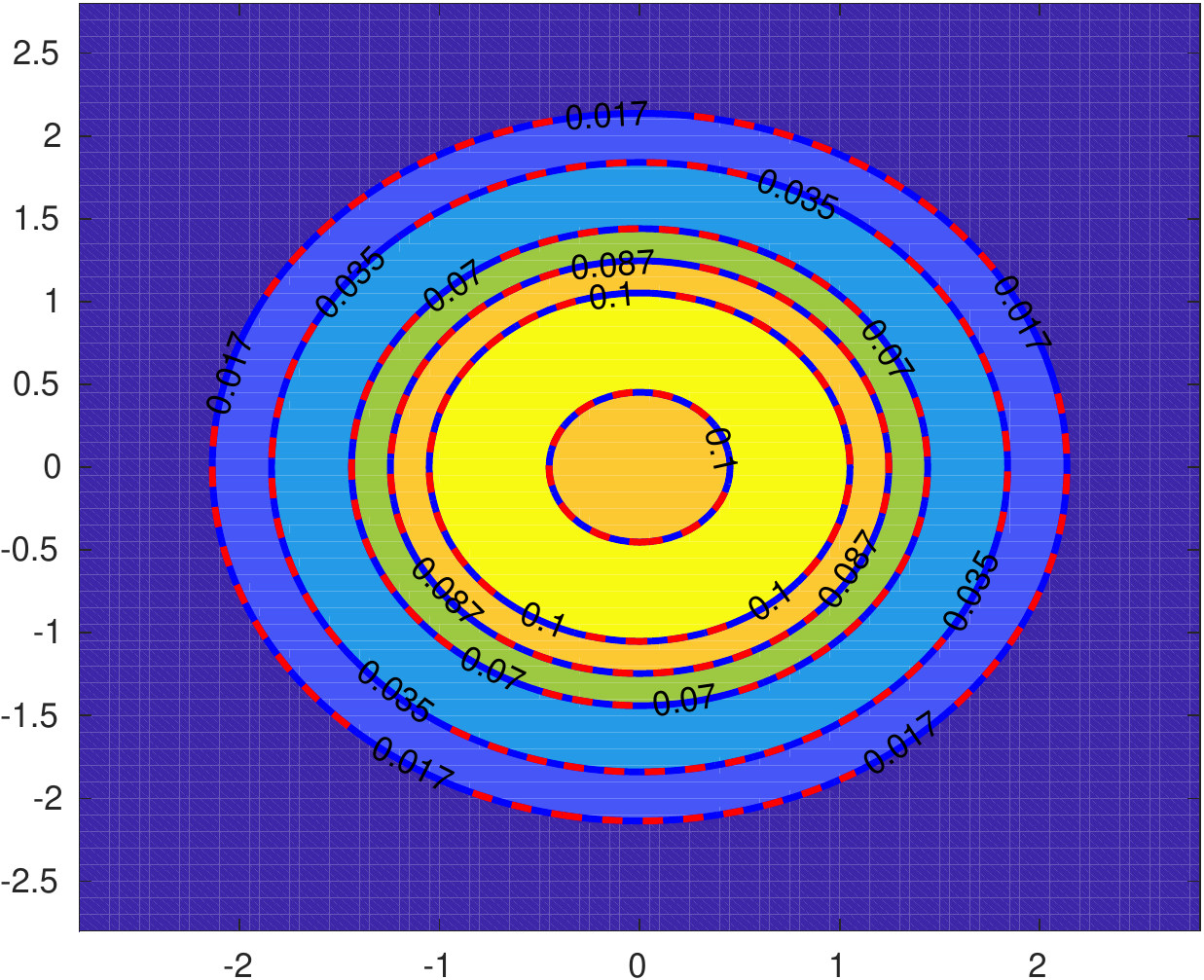}  
}\quad
\subfigure[Initial MDF $\mathcal{I}_2 f^0$\label{fig:ex1_init_2d}]{%
  \includegraphics[width=0.32\textwidth, clip]{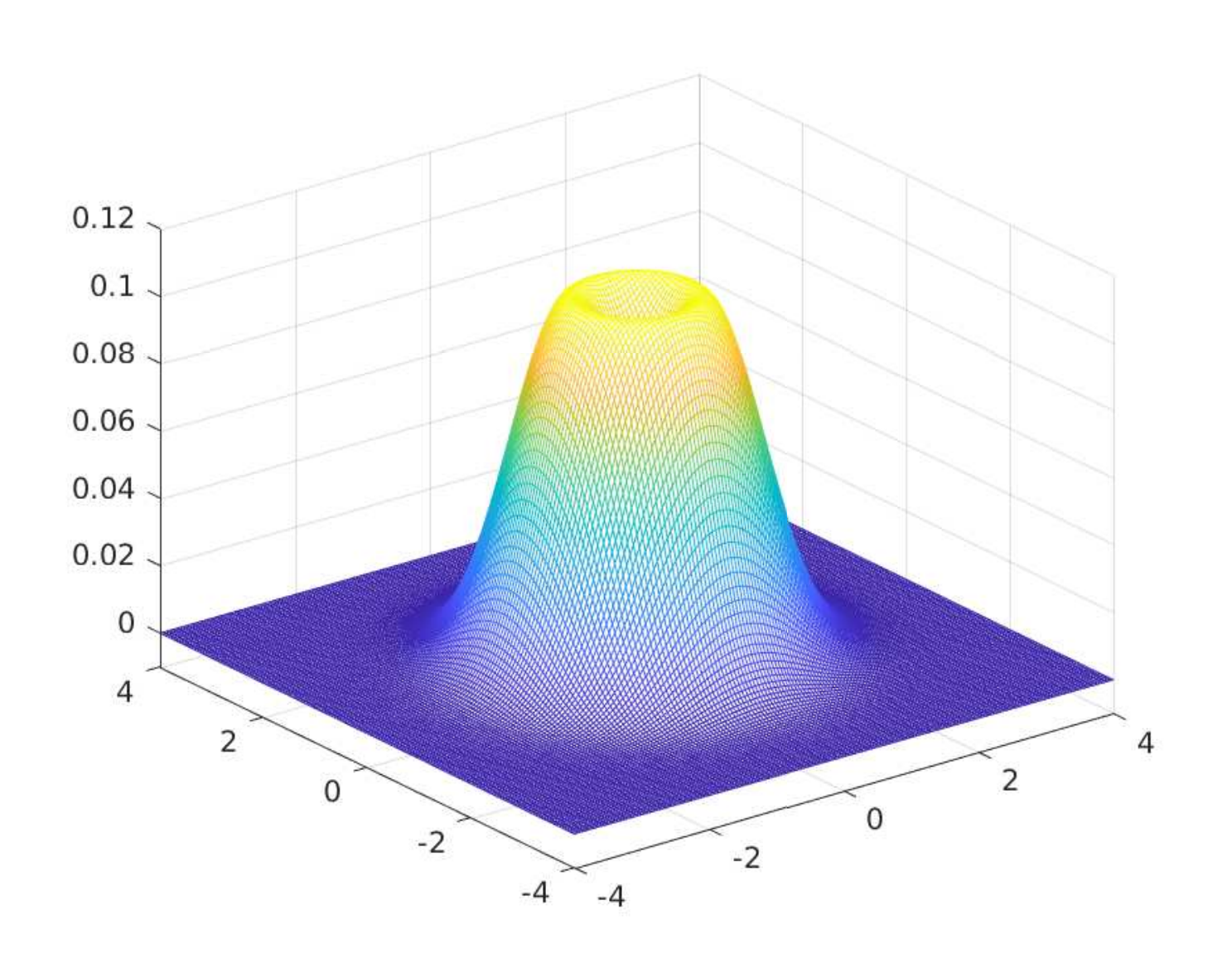}  
}
\caption{Initial marginal distribution functions. In (a), the red line
  corresponds to the exact solution, while black dashed line
  corresponds to $M = 20$ respectively. In (b), the blue solid lines
  correspond to the exact solution, and the red dashed lines
  correspond to the numerical approximation $M = 20$. Figure (c) shows
  only the numerical approximation with $M = 20$.}
\label{fig:ex1_init}
\end{figure}

In Figure \ref{fig:ex1_1d}, the marginal distribution functions
$\mathcal{I}_1F(t)$ at $t = 0.2$, $0.4$ and $0.6$ are shown. Here, $M_0$ is set
as $5$ and $20$.  The marginal distribution functions $\mathcal{I}_2F(t)$ are
plotted in Figures \ref{fig:ex1_2d_M0=5_20} and \ref{fig:ex1_2d_M0=20_20},
respectively for $M_0=5$ and $20$. For $M_0 = 5$, the numerical solution
provides a reasonable approximation, but still has noticeable deviations, while
for $M_0 = 20$, the two solutions match perfectly in all cases.
\begin{figure}[!ht]
\centering
\subfigure[$t=0.2$]{%
  \includegraphics[width=.30\textwidth]{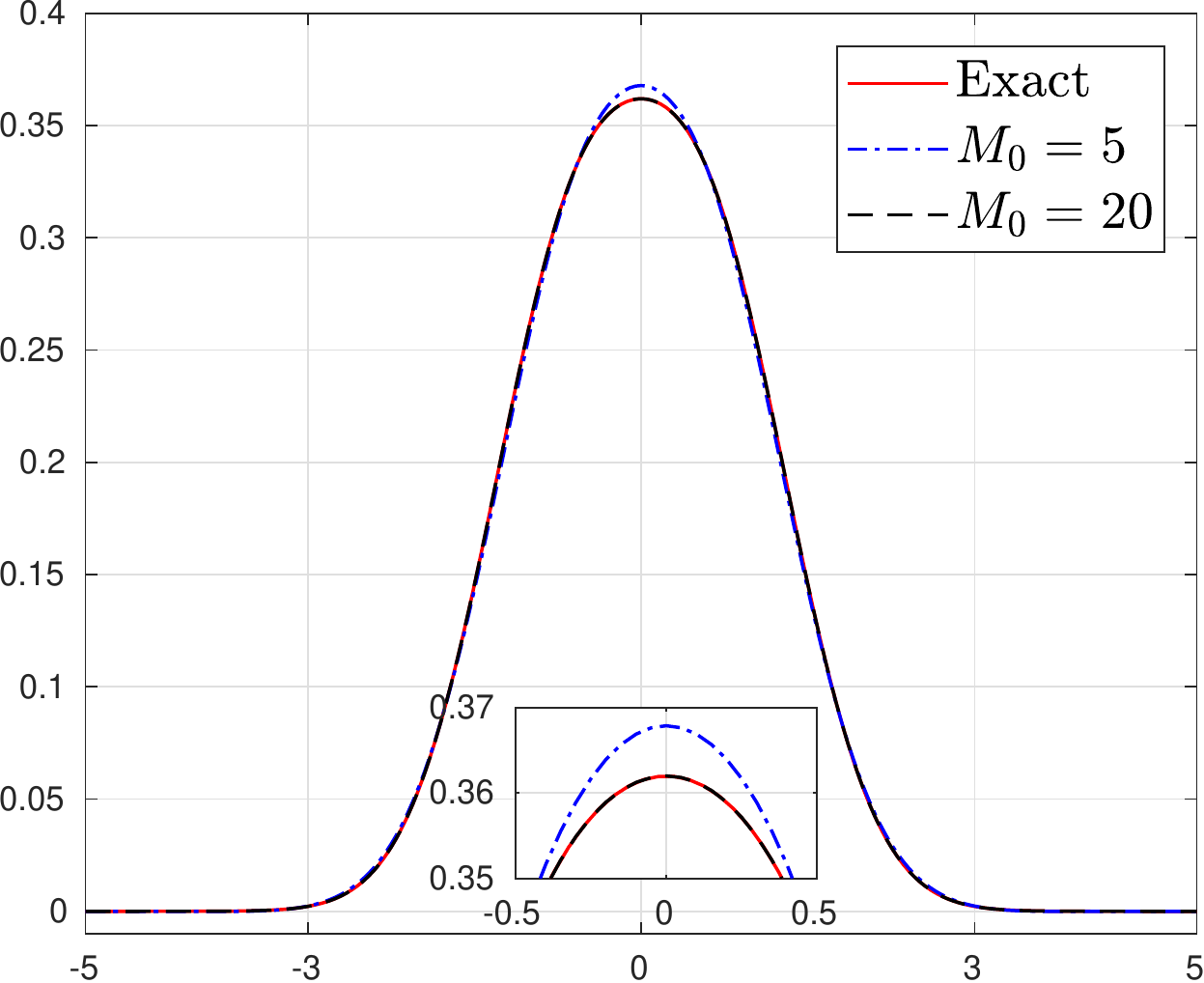}
}\hfill
\subfigure[$t=0.4$]{%
  \includegraphics[width=.30\textwidth]{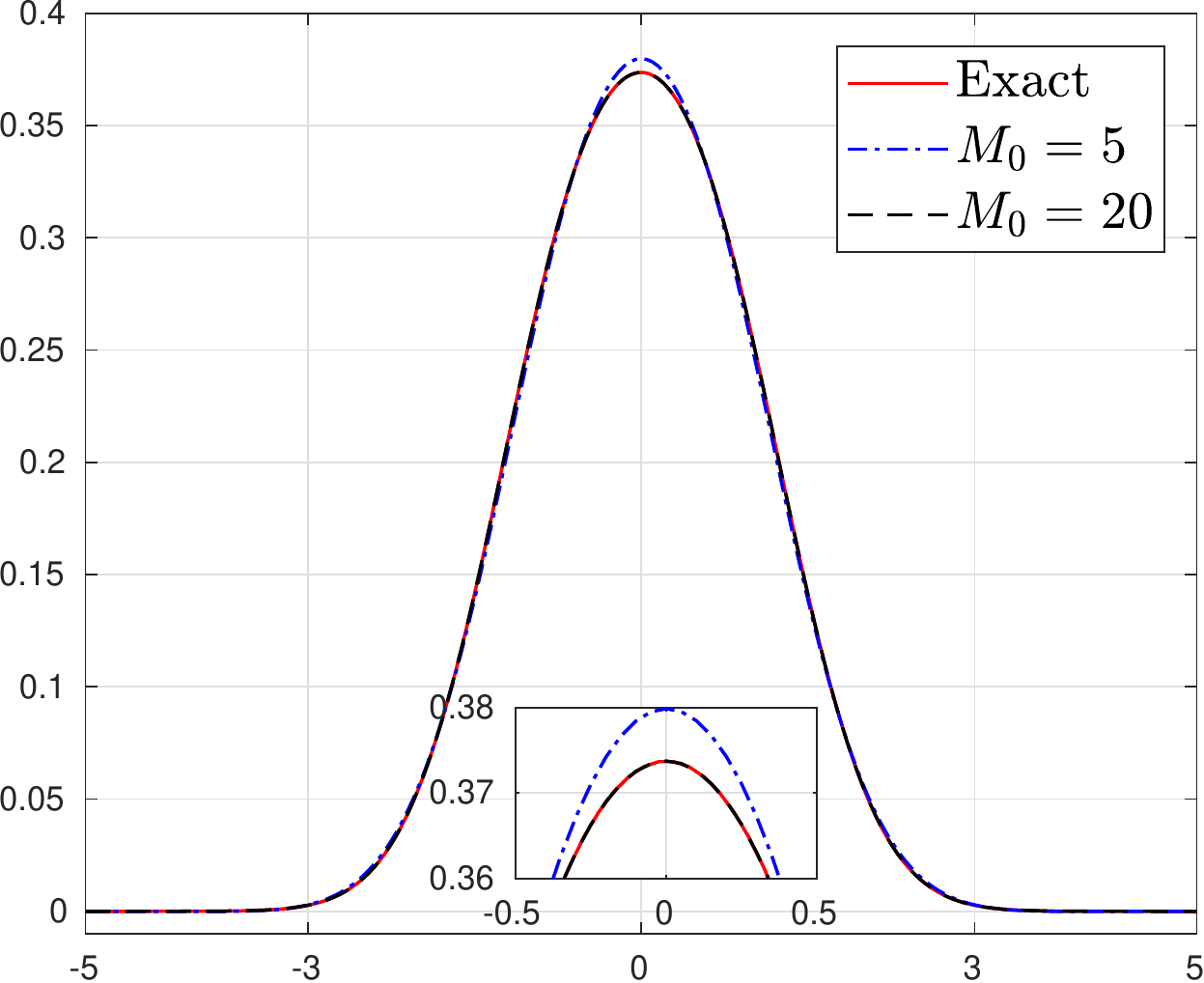}
}\hfill
\subfigure[$t=0.6$]{%
  \includegraphics[width=.30\textwidth]{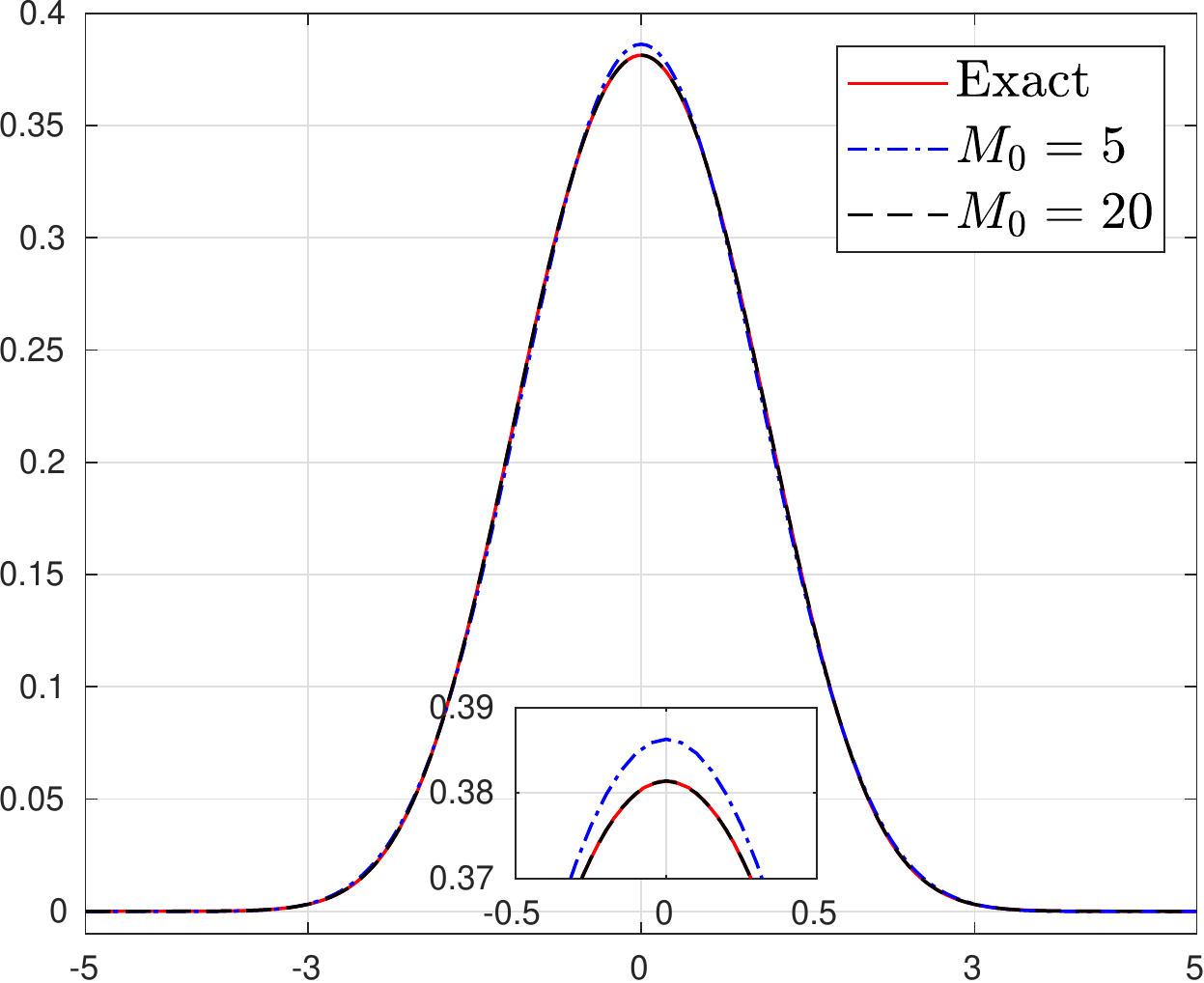}
}
\caption{Marginal distribution functions $\mathcal{I}_1 F(t)$ at different times.}
\label{fig:ex1_1d}
\end{figure}

\begin{figure}[!ht]
\centering
\subfigure[$t=0.2$]{%
  \includegraphics[width=.3\textwidth]{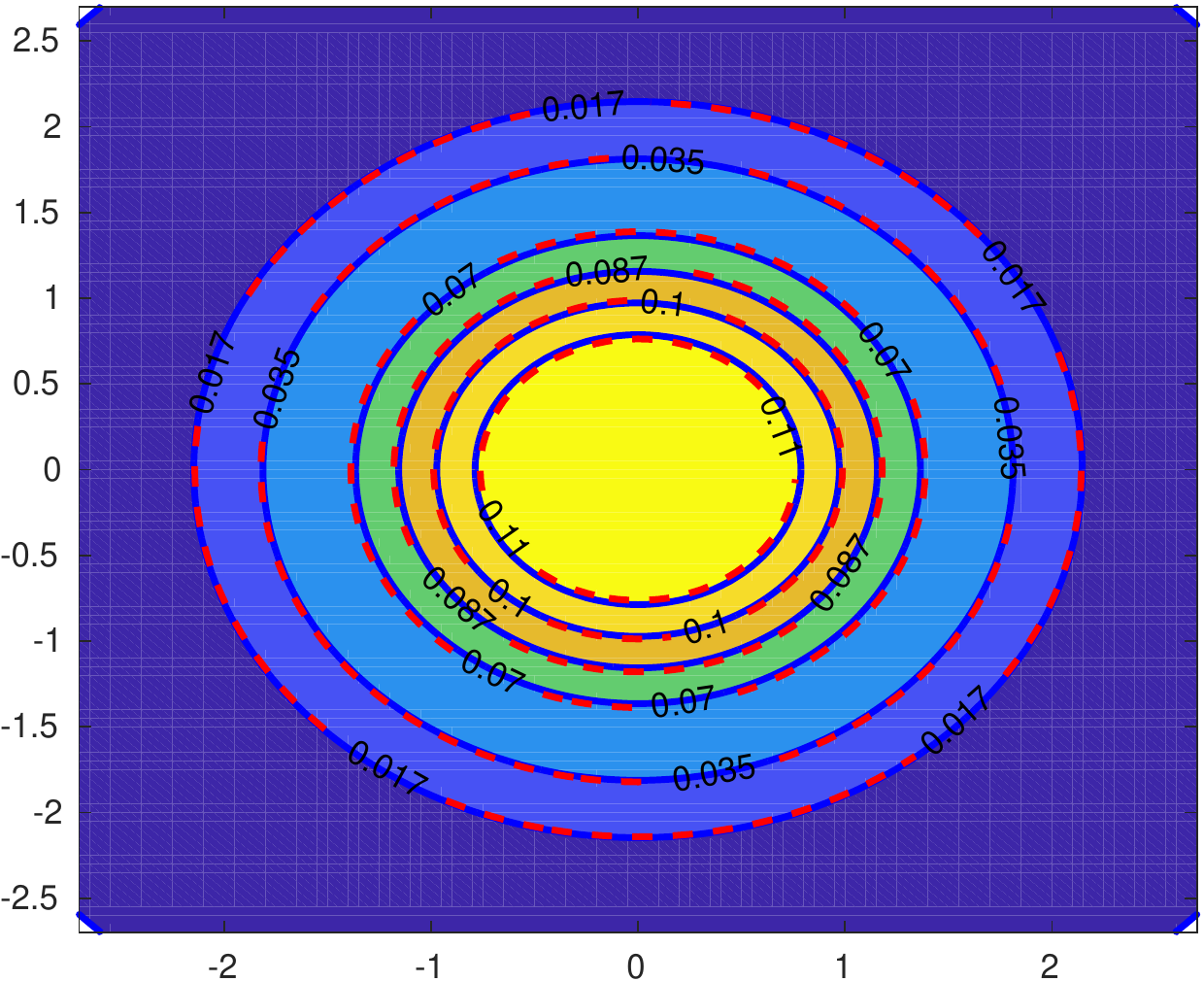}
} \hfill
\subfigure[$t=0.4$]{%
  \includegraphics[width=.3\textwidth]{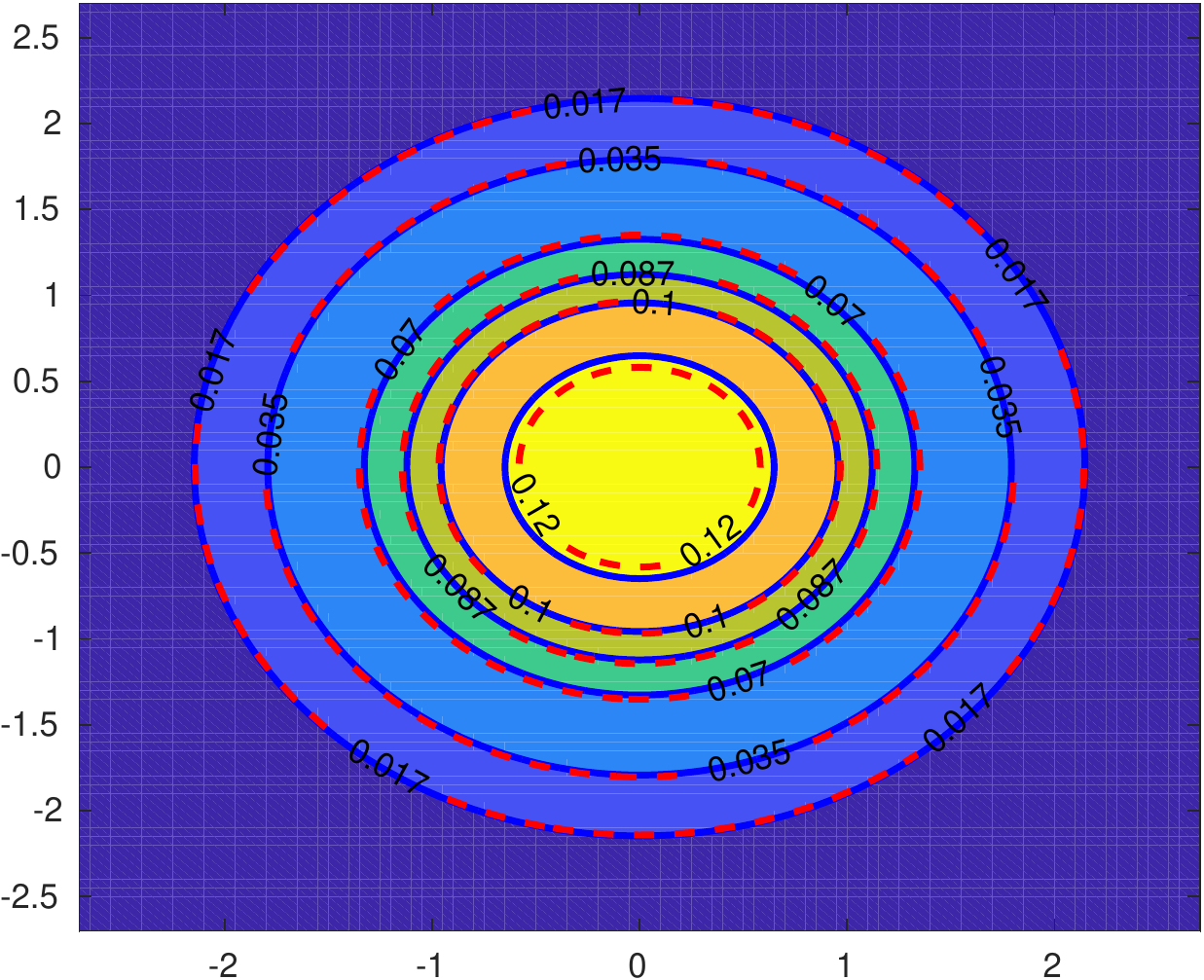}
} \hfill
\subfigure[$t=0.6$]{%
  \includegraphics[width=.3\textwidth]{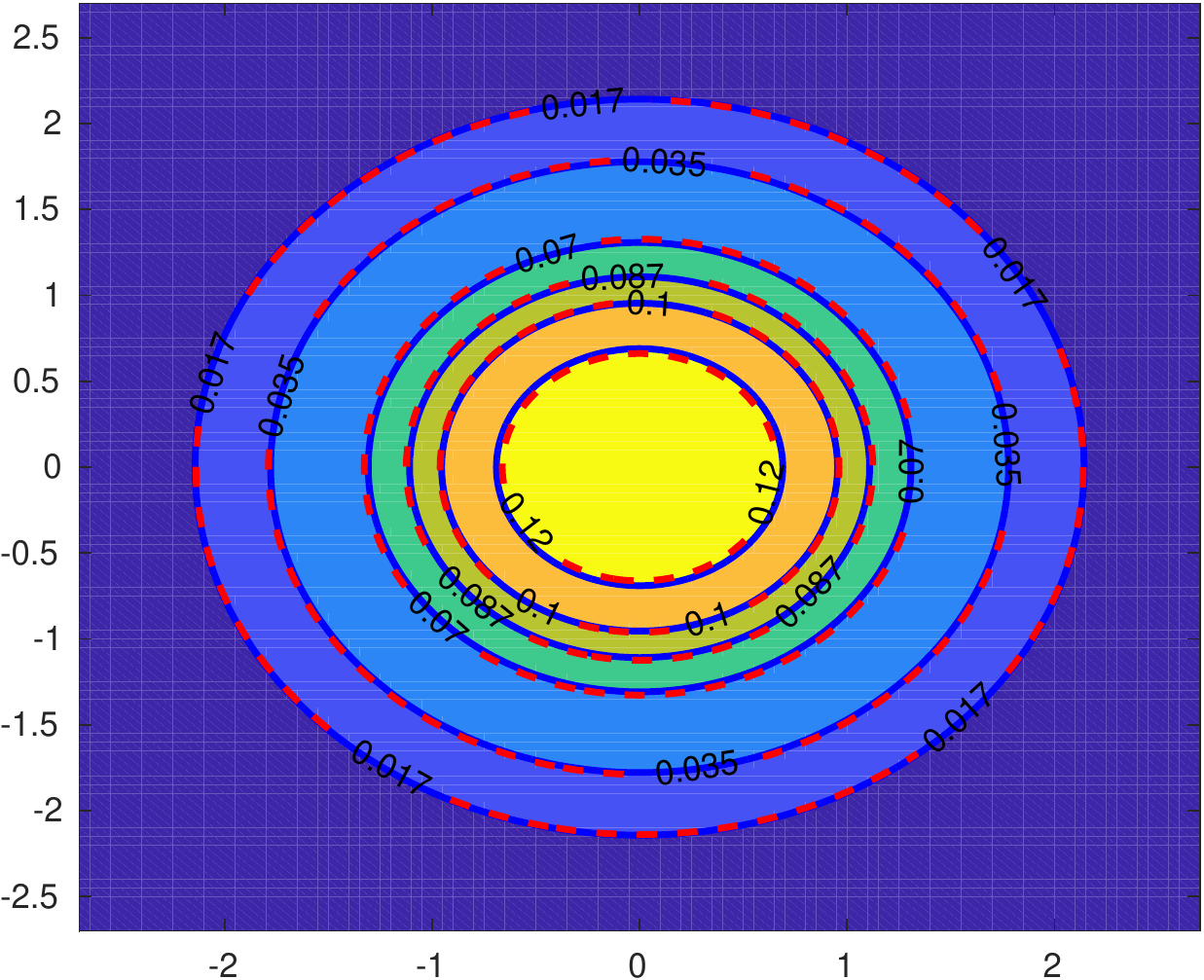}
}
\caption{Comparison of numerical results using $M_0 = 5$ and the exact
  solution. The blue contours and the red dashed contours are
  respectively the results for $M_0 = 5$ and the exact solution.}
\label{fig:ex1_2d_M0=5_20}
\end{figure}

\begin{figure}[!ht]
\centering
\subfigure[$t=0.2$]{%
  \begin{overpic}
  [width=.3\textwidth, clip]{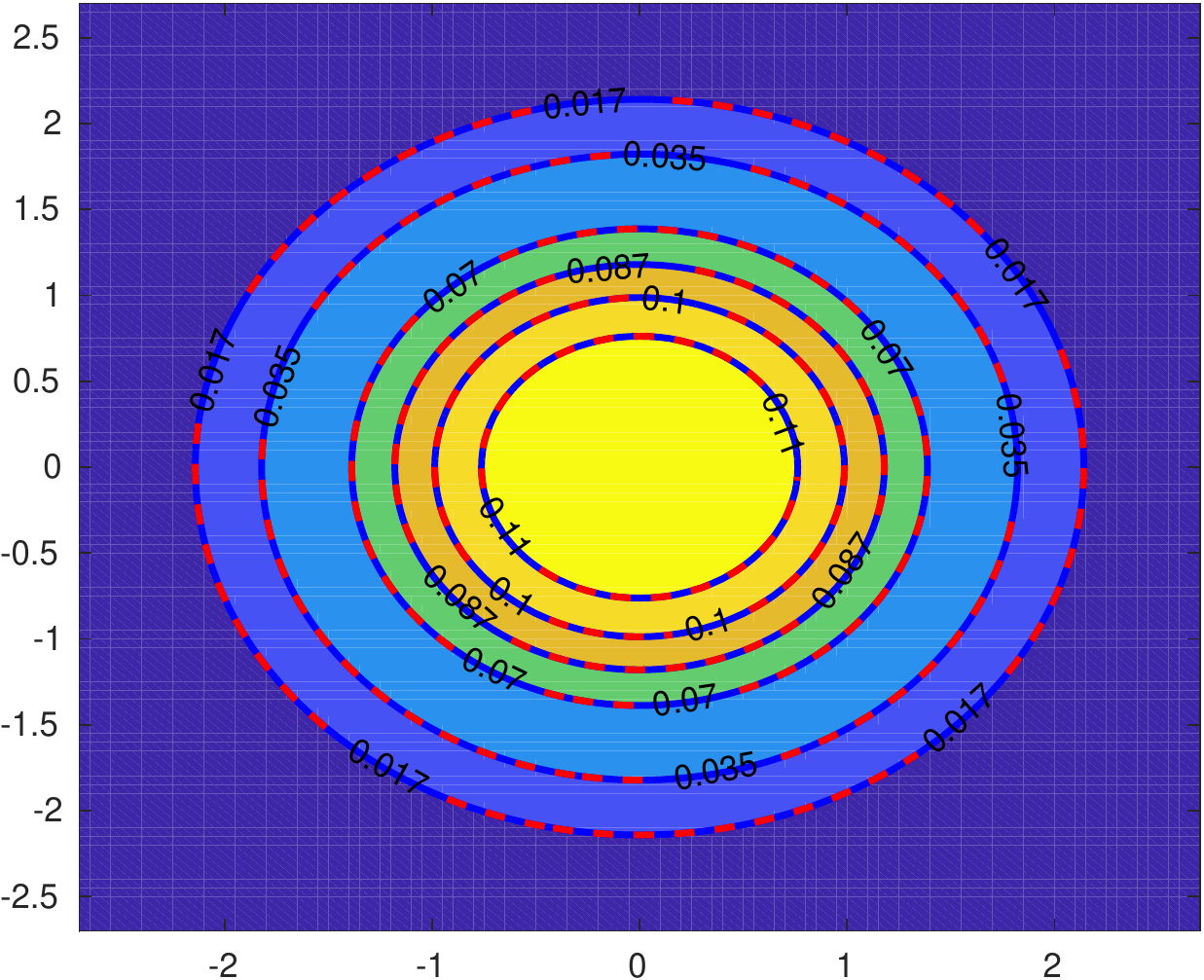}  
  \end{overpic}
  } \quad
\subfigure[$t=0.4$]{%
  \begin{overpic}
    [width=.3\textwidth, clip]{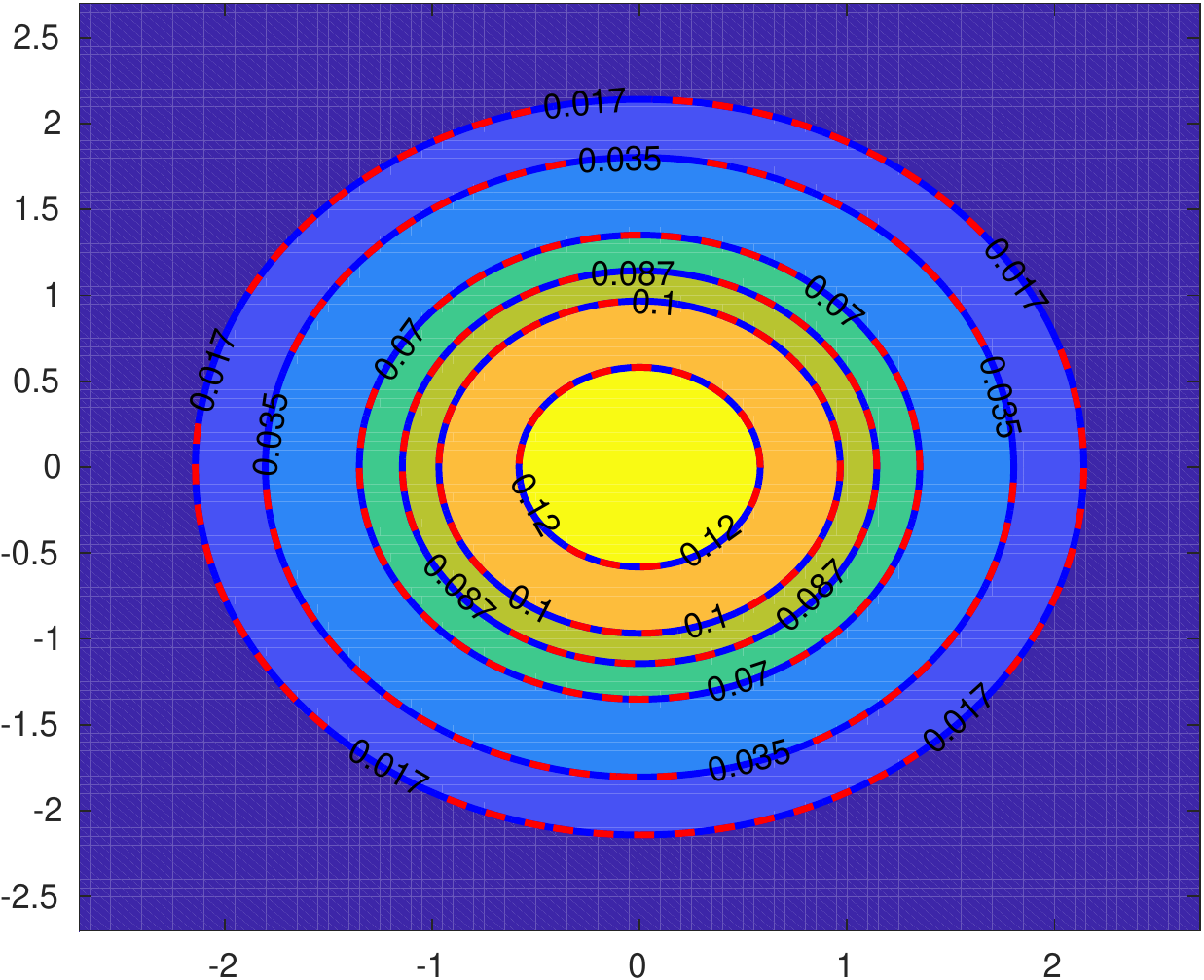}    
  \end{overpic}
} \quad
\subfigure[$t=0.6$]{%
  \begin{overpic}
    [width= .3\textwidth, clip]{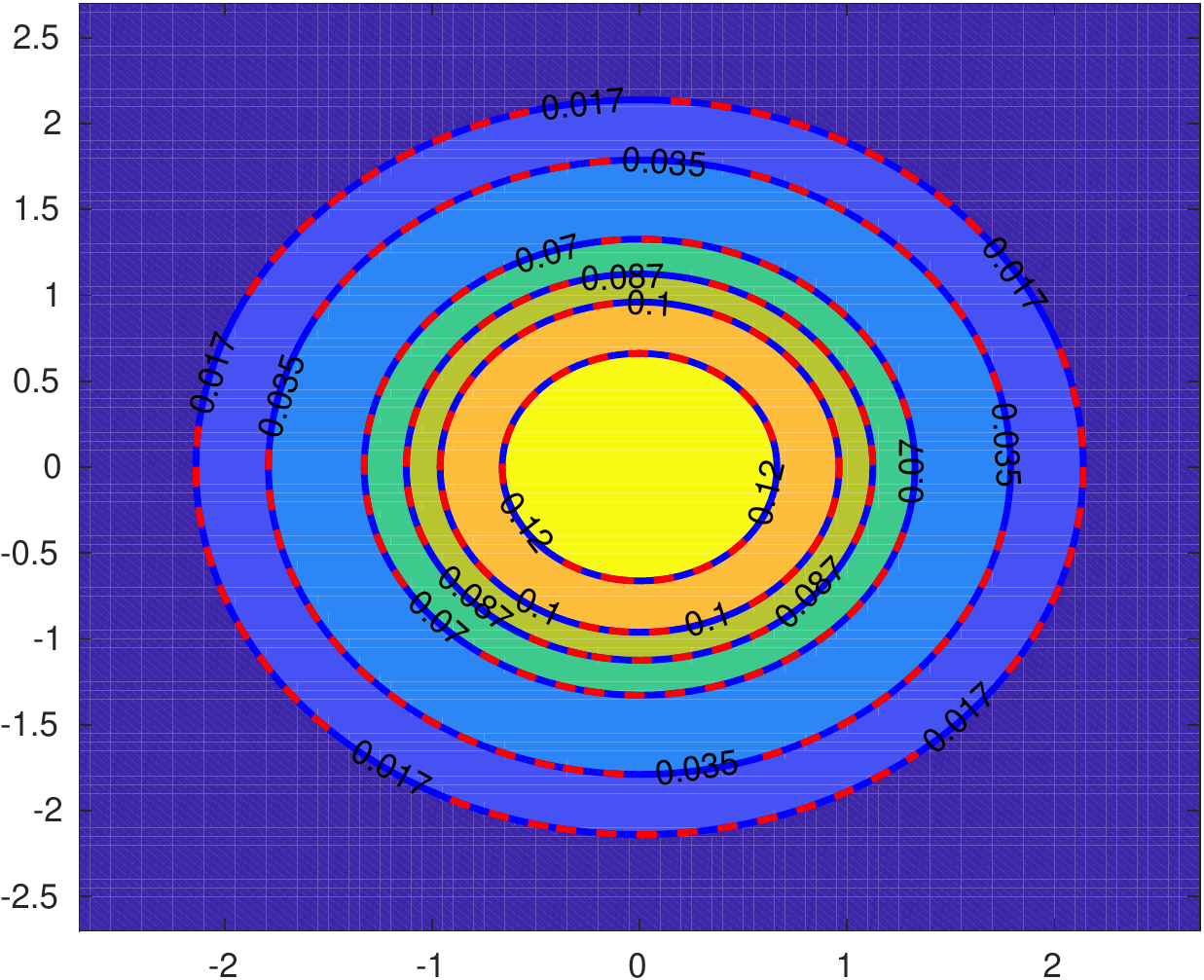}  
  \end{overpic}
}
\caption{Comparison of numerical results using $M_0 = 20$ and the
  exact solution. The blue contours and the red dashed contours are
  respectively the results for $M_0 = 20$ and the exact solution.}
\label{fig:ex1_2d_M0=20_20}
\end{figure}

Now we consider the time evolution of the coefficients. By expanding the exact
solution into Burnett series, we get the exact solution for the coefficients:
\begin{equation}
\label{eq:exact_ex1}
\tilde{F}_{lmn}(t) = \left\{\begin{array}{ll}
  \sqrt{\frac{2\Gamma(n + 3/2)}{\sqrt{\pi} n!}}(1 - n)( 1-  \tau(t))^n, &
  l = m = 0, \quad n\in \bbN, \\[13pt]
  0, & \text{otherwise}.
\end{array} \right.
\end{equation}
Due to the symmetry of the distribution function, the coefficients
$\tilde{F}_{lmn}$ are nonzero for any $t$ only when both of $l$ and $m$ are
zero. From \eqref{eq:exact_ex1}, we see that $\tilde{F}_{000} = 1$ and
$\tilde{F}_{001} = 0$ for any $t$. Hence we will focus on the coefficients
$\tilde{F}_{00n}, n = 2,\cdots,5$. For Maxwell molecules, the discrete kernel
$A_{lmn}^{l_1m_1n_1, l_2m_2n_2}$ is nonzero only when $l + 2n = l_1 + 2 n_1 +
l_2 + 2n_2$. Therefore, for any $M \geqslant M_0 \geqslant 10$, the numerical
results for these coefficients are exactly the same (regardless of round-off
errors), and we just show the results for $M_0 = M =20$ here. Figure
\ref{fig:ex1_moments} gives the comparison between the numerical solution and
the exact solution for these coefficients. In all plots, the two lines almost
coincide with each other.

\begin{figure}[!ht]
\centering
\subfigure[$\tilde{F}_{002}(t)$]{%
  \includegraphics[width=.4\textwidth]{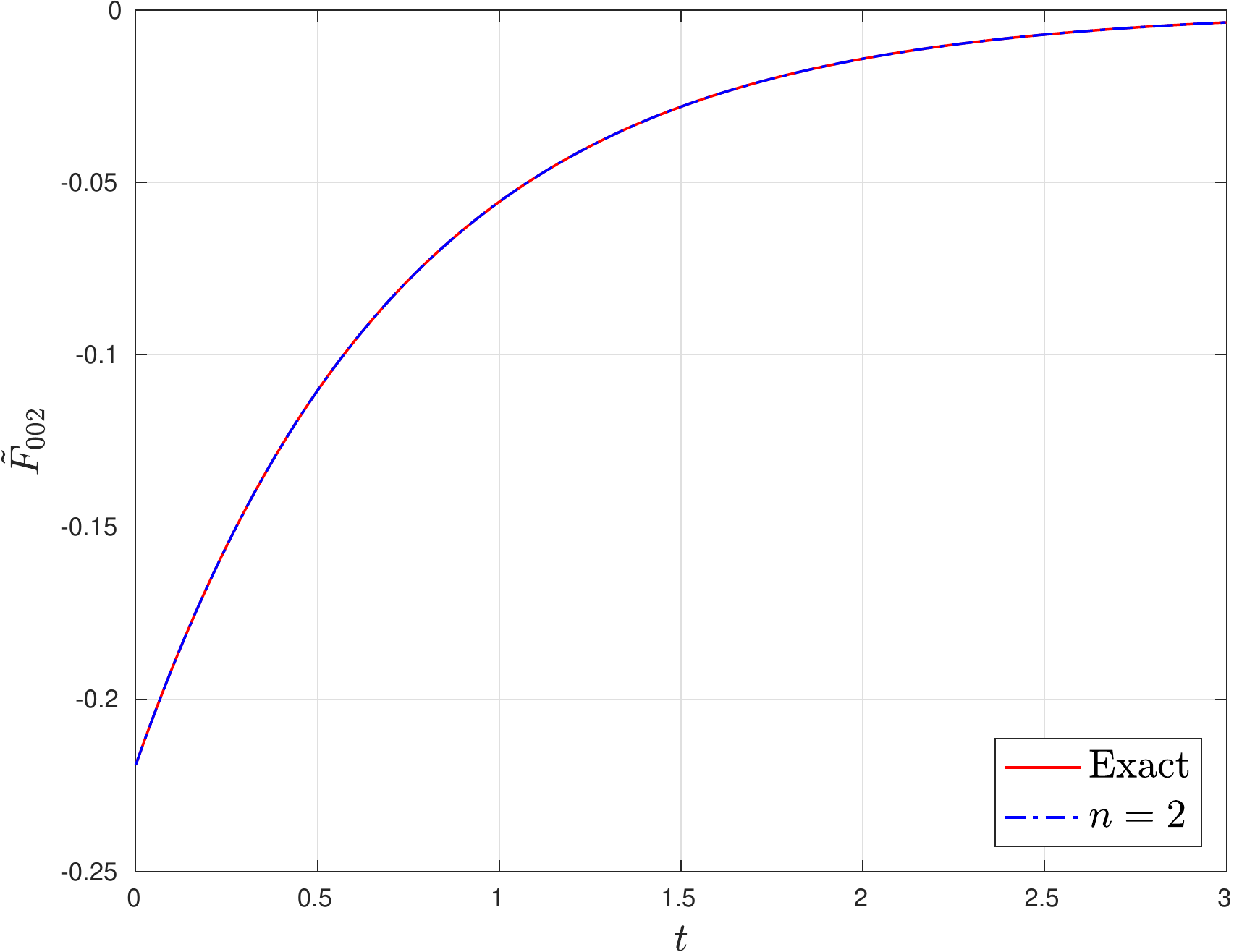}
} \hspace{20pt}
\subfigure[$\tilde{F}_{003}(t)$]{%
  \includegraphics[width=.4\textwidth]{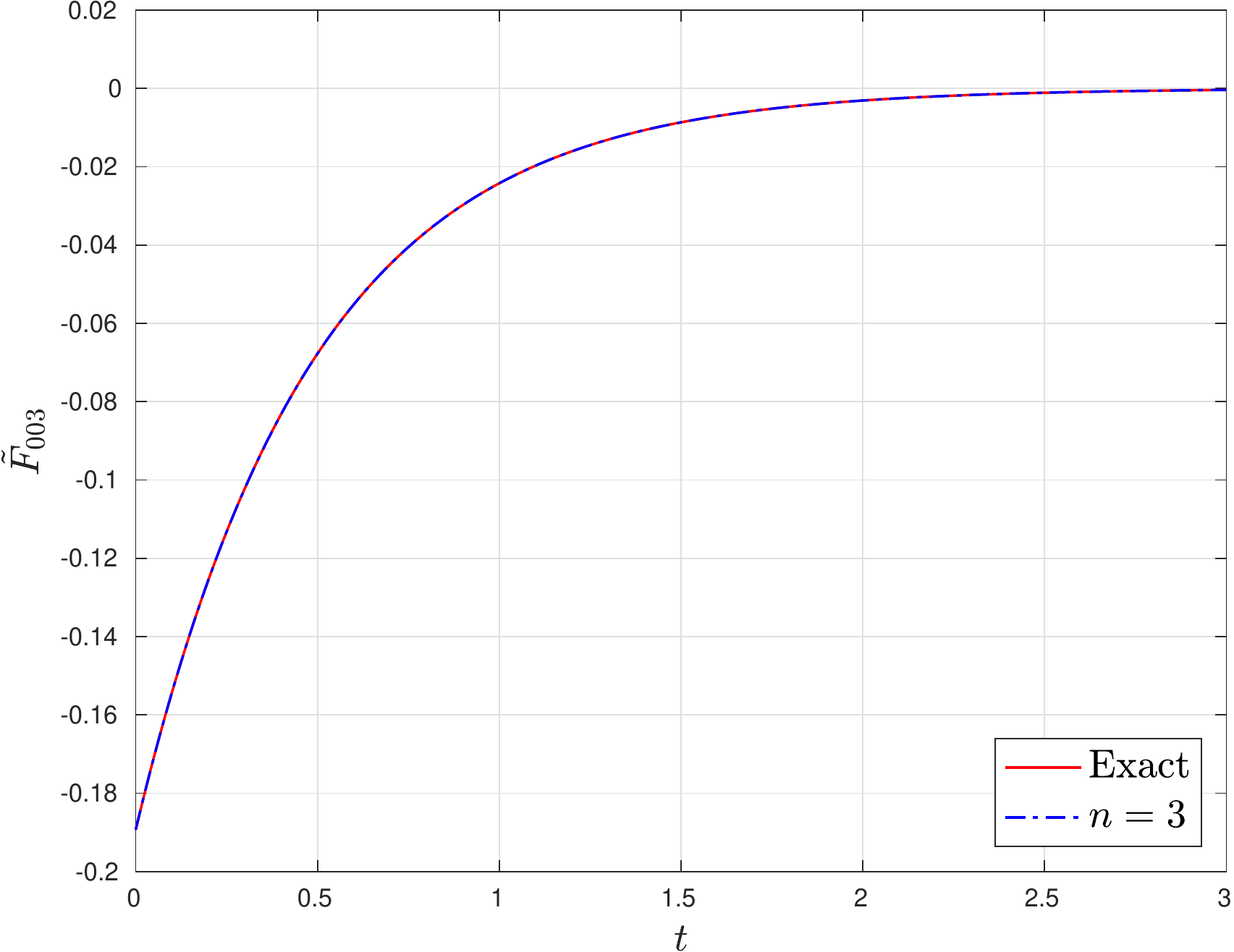}
} \\
\subfigure[$\tilde{F}_{004}(t)$]{%
  \includegraphics[width=.4\textwidth]{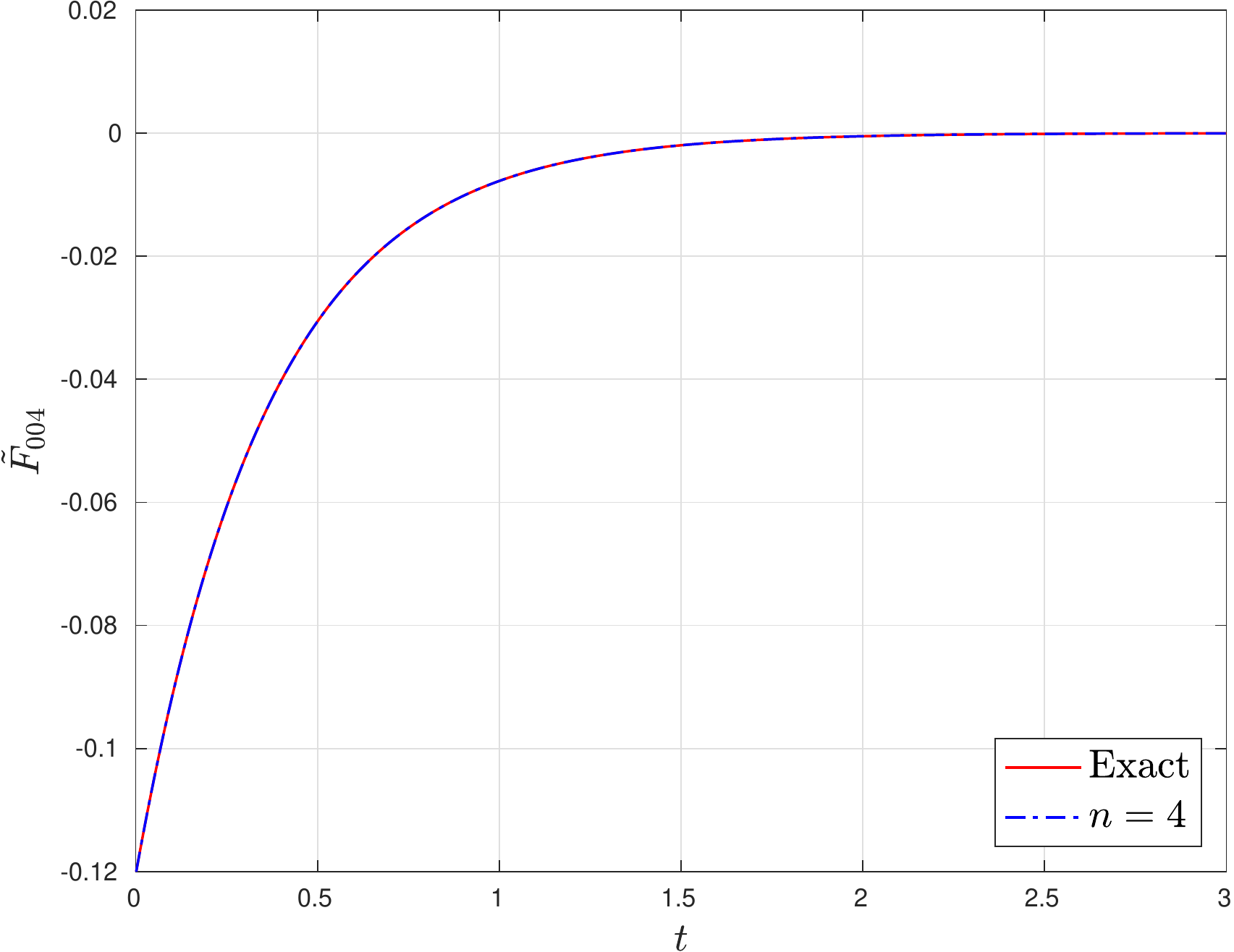}
} \hspace{20pt}
\subfigure[$\tilde{F}_{005}(t)$]{%
  \includegraphics[width=.4\textwidth]{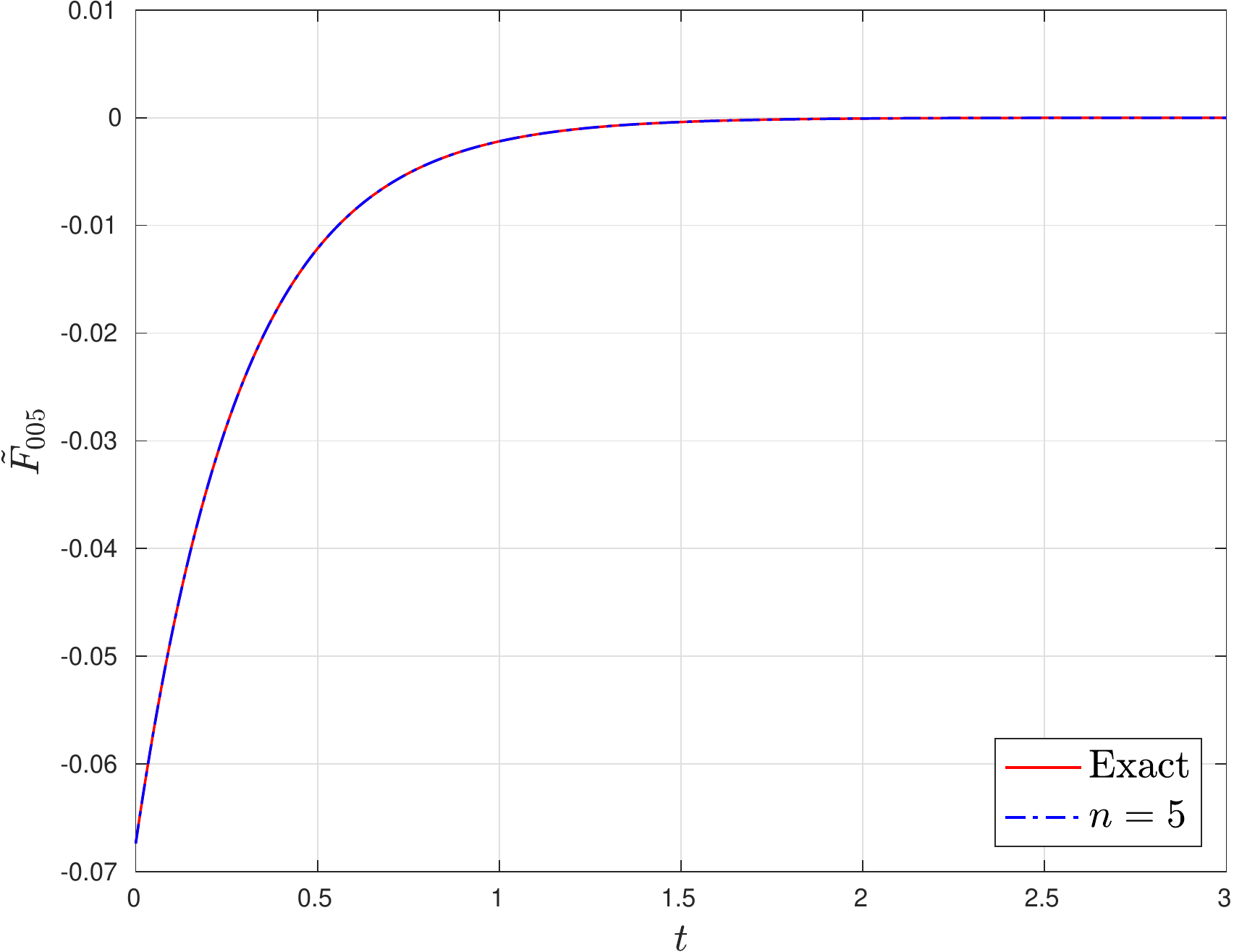}
}
\caption{The evolution of the coefficients. The red lines correspond to the
  reference solution, and the blue dashed lines correspond to the numerical
  solution.}
\label{fig:ex1_moments}
\end{figure}

\subsection{Quadruple-Gaussian initial data}
In this example, we perform the numerical test for the hard potential
case where the power index $\eta$ equals $10$. The initial distribution
function is
\begin{equation}
  \begin{aligned}
    f^0(\bv) = \frac{1}{4\pi^{3/2}} \left[ \exp \Big( -\frac{(v_1 +
        u)^2 + v_2^2 + v_3^2 }{2\theta}\Big) +
      \exp \Big( -\frac{(v_1 - u)^2 + v_2^2 + v_3^2}{2\theta} \Big)\right. \\
    \left.  + \exp \Big( -\frac{v_1^2 + (v_2+u)^2 + v_3^2}{2\theta}
      \Big) + \exp \Big( -\frac{v_1^2 + (v_2- u)^2 + v_3^2}{2\theta}
      \Big) \right],
  \end{aligned}
\end{equation}
where $u = \sqrt{2}$ and $\theta = 1/3$. In all our numerical tests,
we use $M = 40$, which gives a good approximation of the initial
distribution function (see Figure \ref{fig:ex2_init}).

\begin{figure}[!ht]
\centering
\subfigure[Initial MDF $\mathcal{I}_1 f^0$\label{fig:ex2_init_1d}]{%
  \begin{overpic}[width=0.28\textwidth, clip]{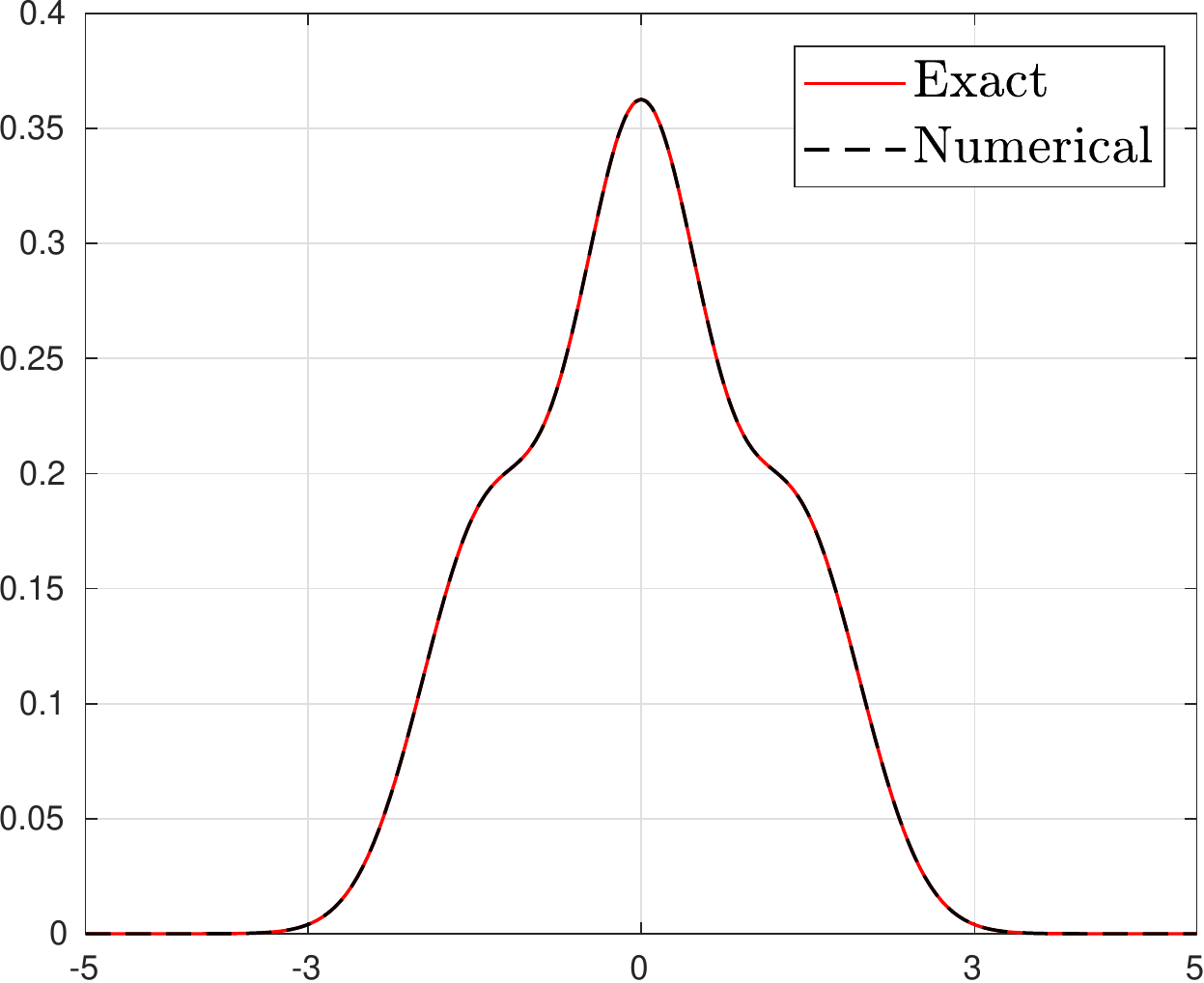}    
  \end{overpic}
}\quad
\subfigure[Contours of $\mathcal{I}_2 f^0$\label{fig:ex2_init_2d_contour}]{%
  \begin{overpic}
    [width=0.3\textwidth, clip]{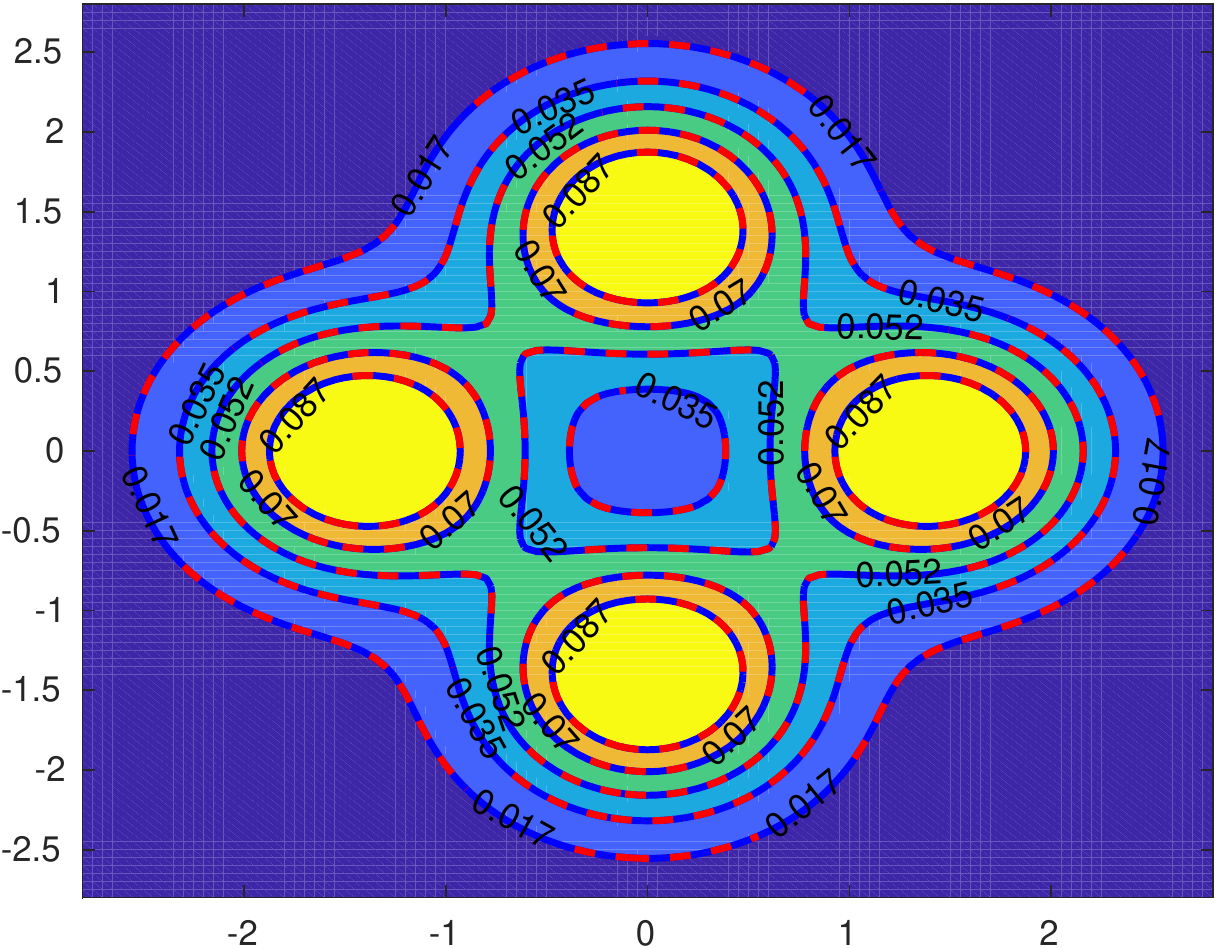}  
  \end{overpic}
}\quad
\subfigure[Initial MDF $\mathcal{I}_2 f^0$\label{fig:ex2_init_2d}]{%
  \begin{overpic}
    [width=0.32\textwidth, clip]{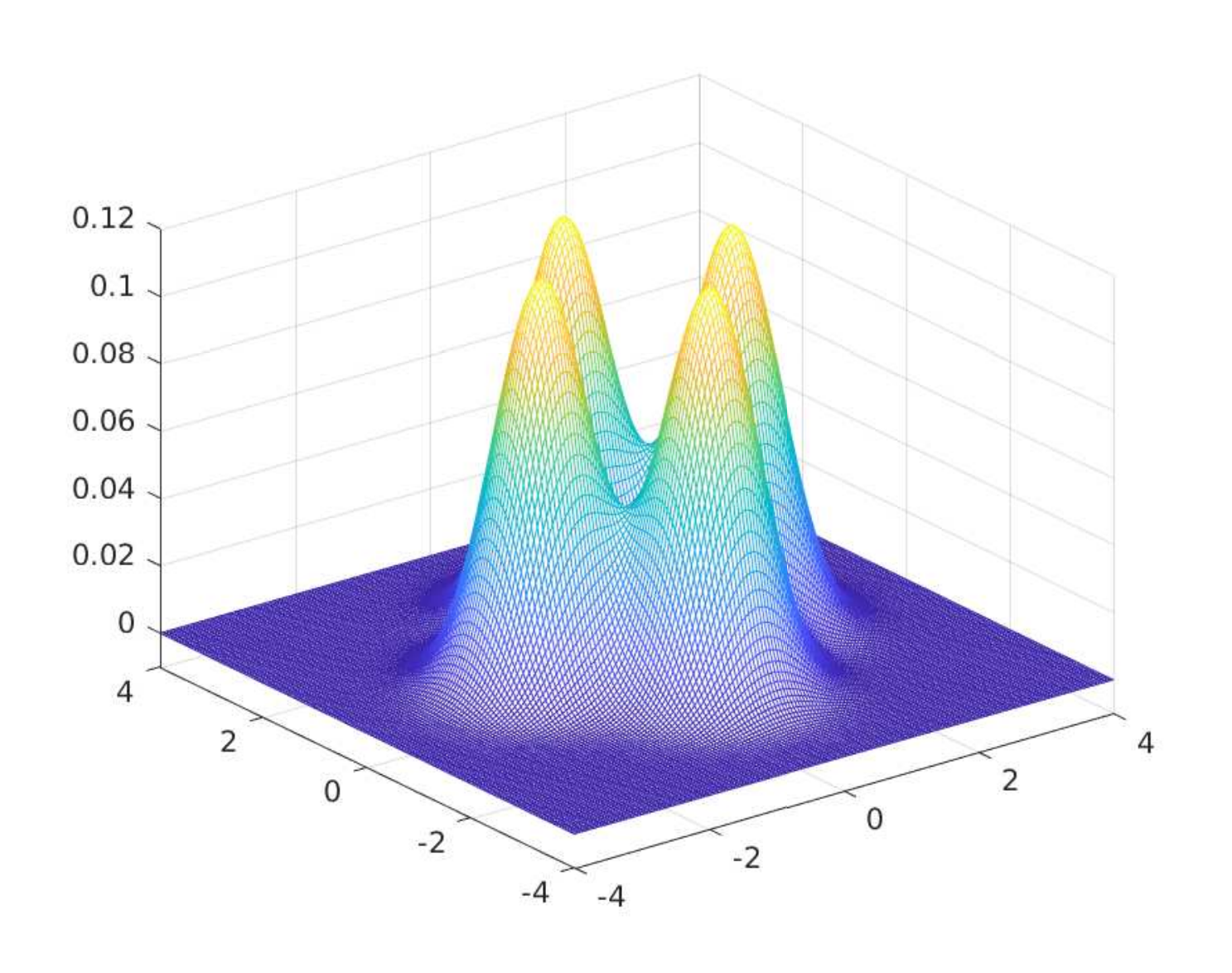}  
  \end{overpic}
}
\caption{Initial marginal distribution functions.  In (a), the red
  line corresponds to the exact solution, while black dashed line
  corresponds to $M = 40$ respectively. In (b), the blue solid lines
  correspond to the exact solution, and the red dashed lines
  correspond to the numerical approximation $M = 40$.  Figure (c)
  shows only the numerical approximation $M=40$.}
\label{fig:ex2_init}
\end{figure}

For this example, we set the numerical result with $M_0=20$ as the
reference solution.  The numerical results for $M_0=5$, $10$, $15$ are
given respectively in Figure \ref{fig:ex2_2d_M0=5_40},
\ref{fig:ex2_2d_M0=10_40}, and \ref{fig:ex2_2d_M0=15_40}. For each
$M_0$, the marginal distribution functions $\mathcal{I}_2F(t)$ at
$t=0.1$, $0.2$ and $0.3$ are shown.

Due to the high nonequilibrium of this example, when $M_0 = 5$ and
$10$, the ``size'' of the quadratic part in the collision term is too
small to describe the evolution of the distribution function, while
the numerical results for $M_0 = 15$ and $M_0 = 20$ agree well with
each other (except the central area for $t=0.1$, where the
distribution function is very flat). This indicates the observation of
numerical convergence, meaning that $M_0=15$ is sufficient to describe
the evolution of the distribution function. This example is a harder
version of the bi-Gaussian initial data used in \cite{QuadraticCol}
for Hermite basis functions, and therefore requires more degrees of
freedom to give satisfactory numerical results. However, as will be
shown later, by using Burnett basis functions, the computational cost
for $M_0 = 20$ is even smaller than the computational cost for
$M_0 = 15$ using Hermite basis functions, even if the results are
essentially identical.

\begin{figure}[!ht]
\centering
\subfigure[$t=0.1$]{%
  \includegraphics[width=.3\textwidth]{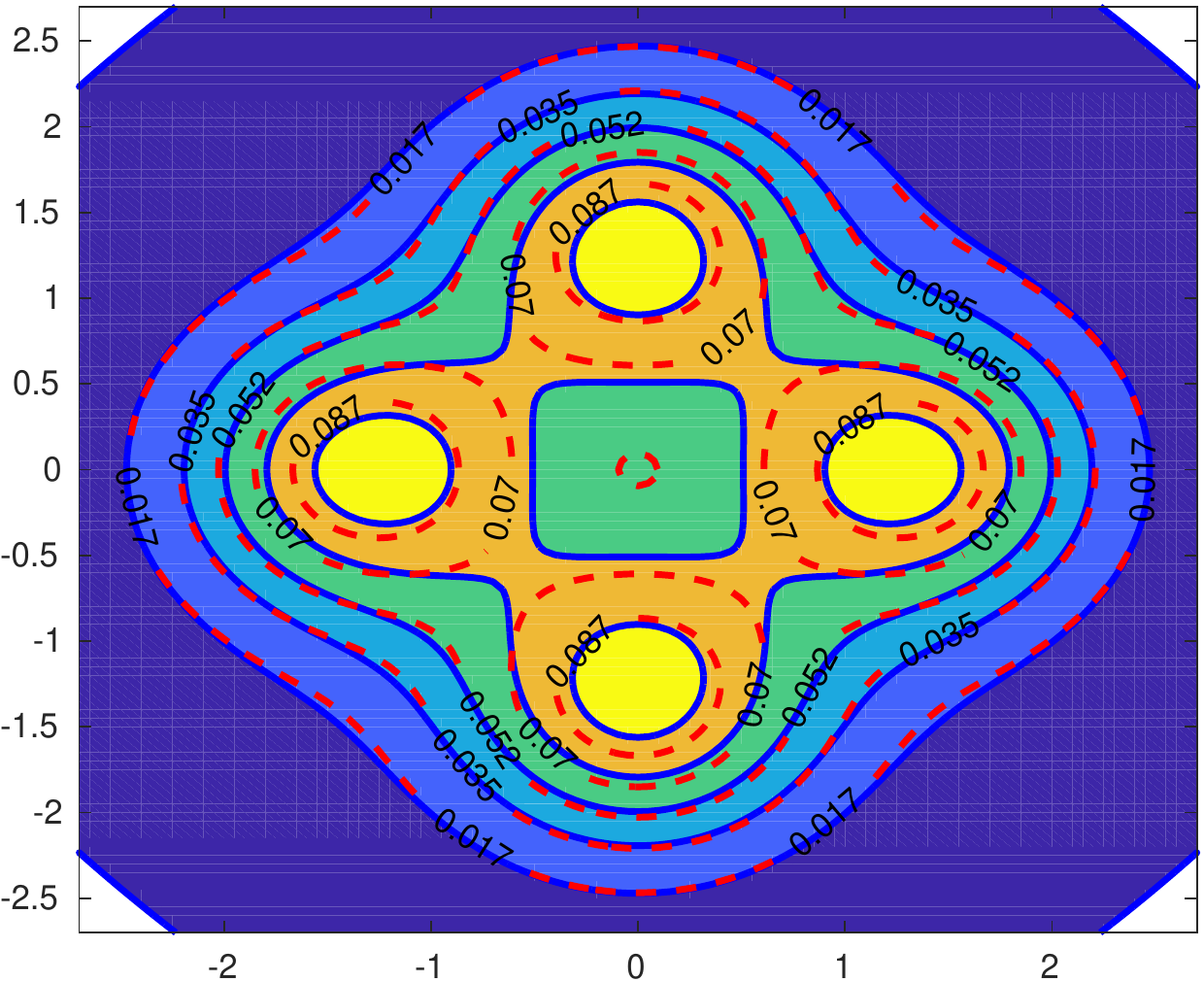}
} \hfill
\subfigure[$t=0.2$]{%
  \includegraphics[width=.3\textwidth]{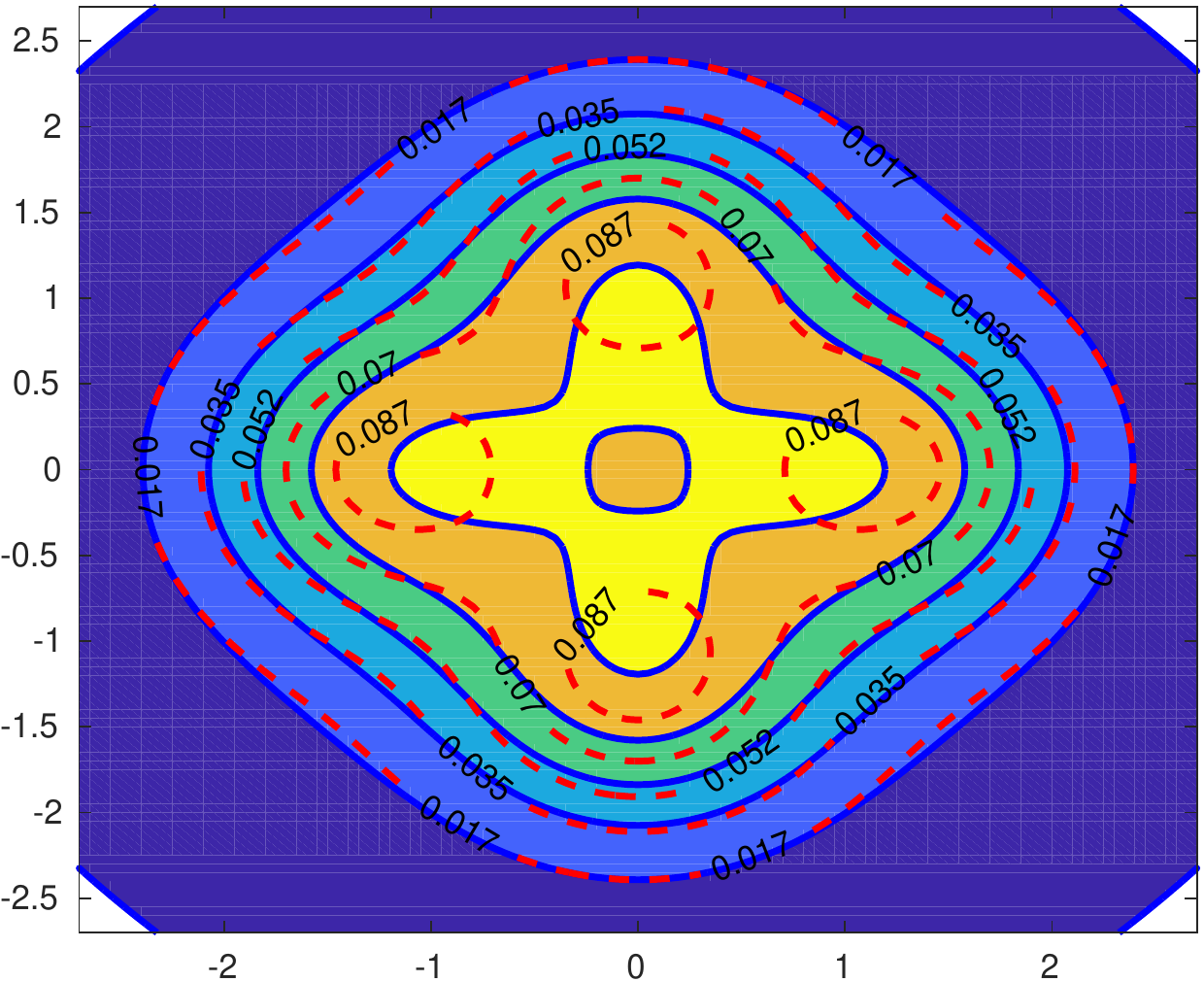}
} \hfill
\subfigure[$t=0.3$]{%
  \includegraphics[width=.3\textwidth]{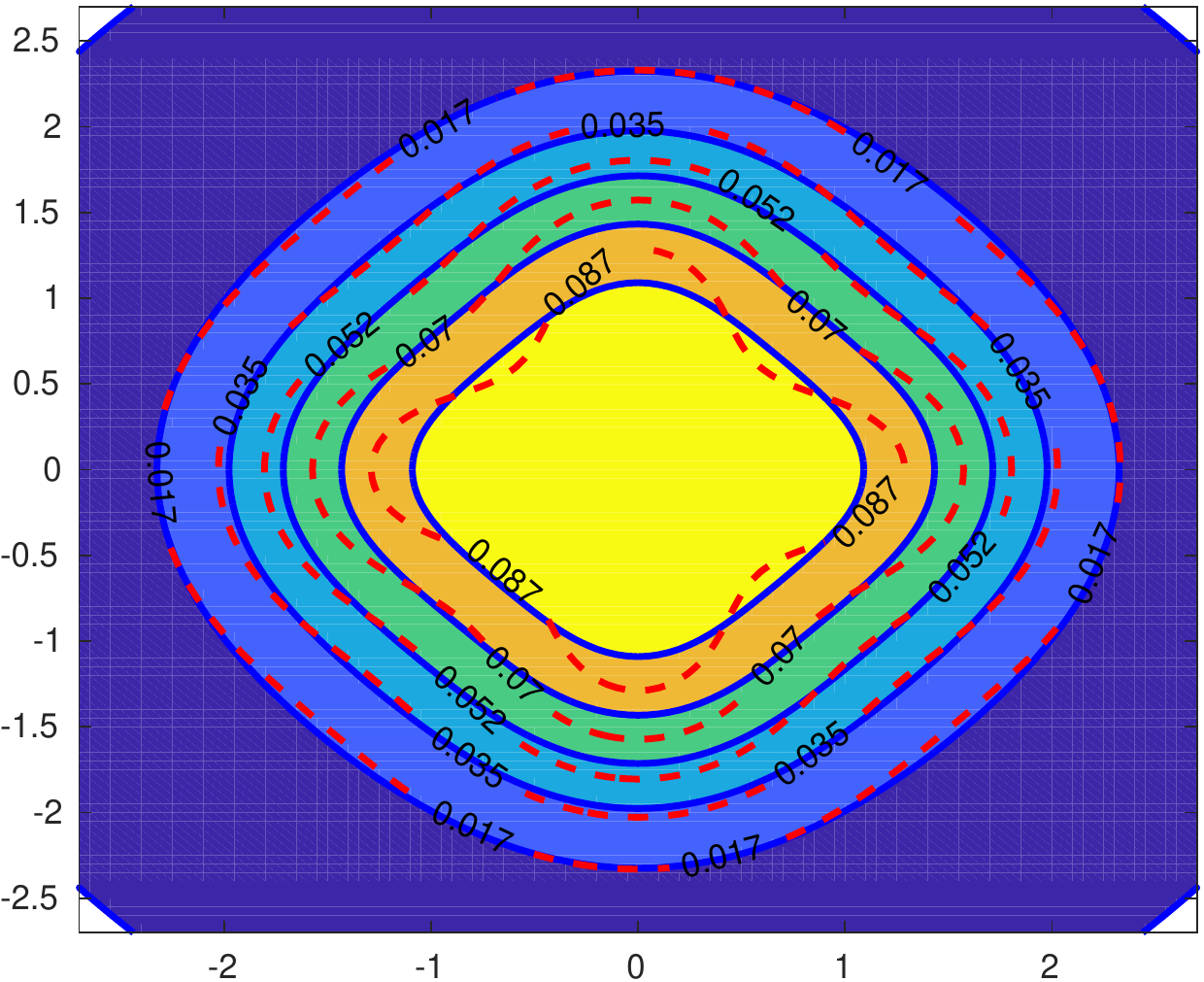}
}
\caption{Comparison of numerical results using $M_0 = 5$ and
  $M_0 = 20$. The blue contours and the red dashed contours are
  respectively the results for $M_0 = 5$ and $M_0 = 20$.}
\label{fig:ex2_2d_M0=5_40}
\end{figure}

\begin{figure}[!ht]
\centering
\subfigure[$t=0.1$]{%
  \begin{overpic}
  [width=.3\textwidth, clip]{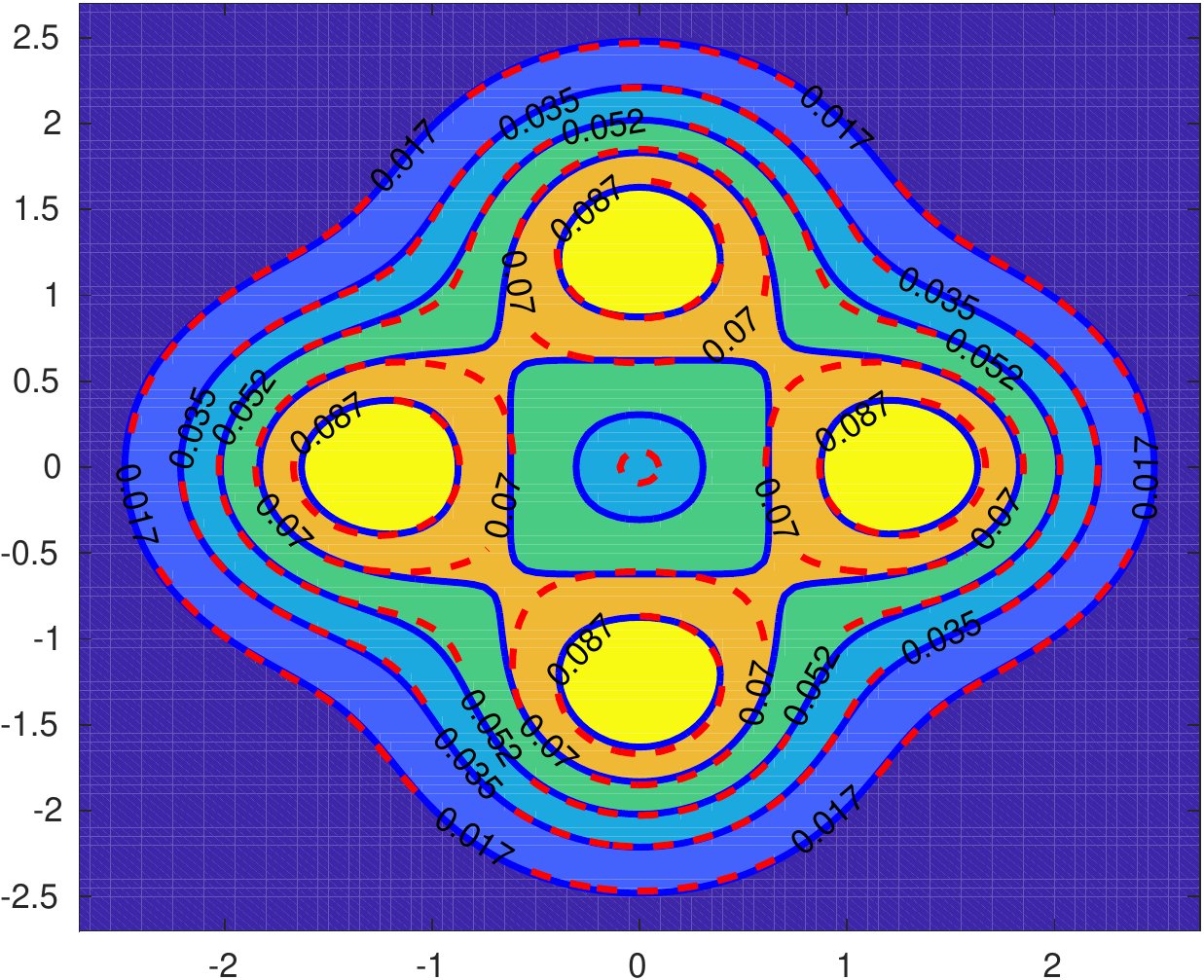}  
  \end{overpic}
  } \quad
\subfigure[$t=0.2$]{%
  \begin{overpic}
    [width=.3\textwidth, clip]{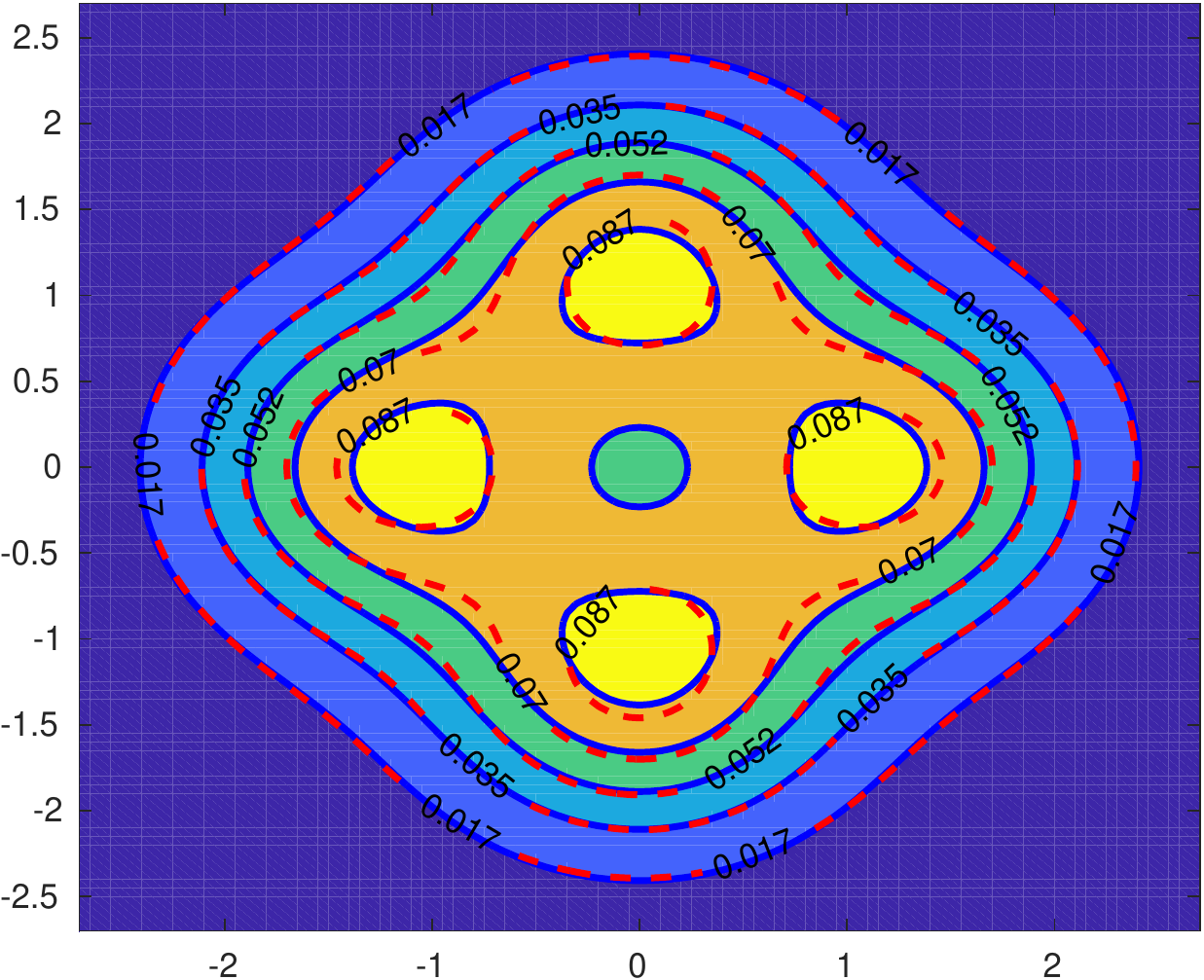}    
  \end{overpic}
} \quad
\subfigure[$t=0.3$]{%
  \begin{overpic}
    [width= .3\textwidth, clip]{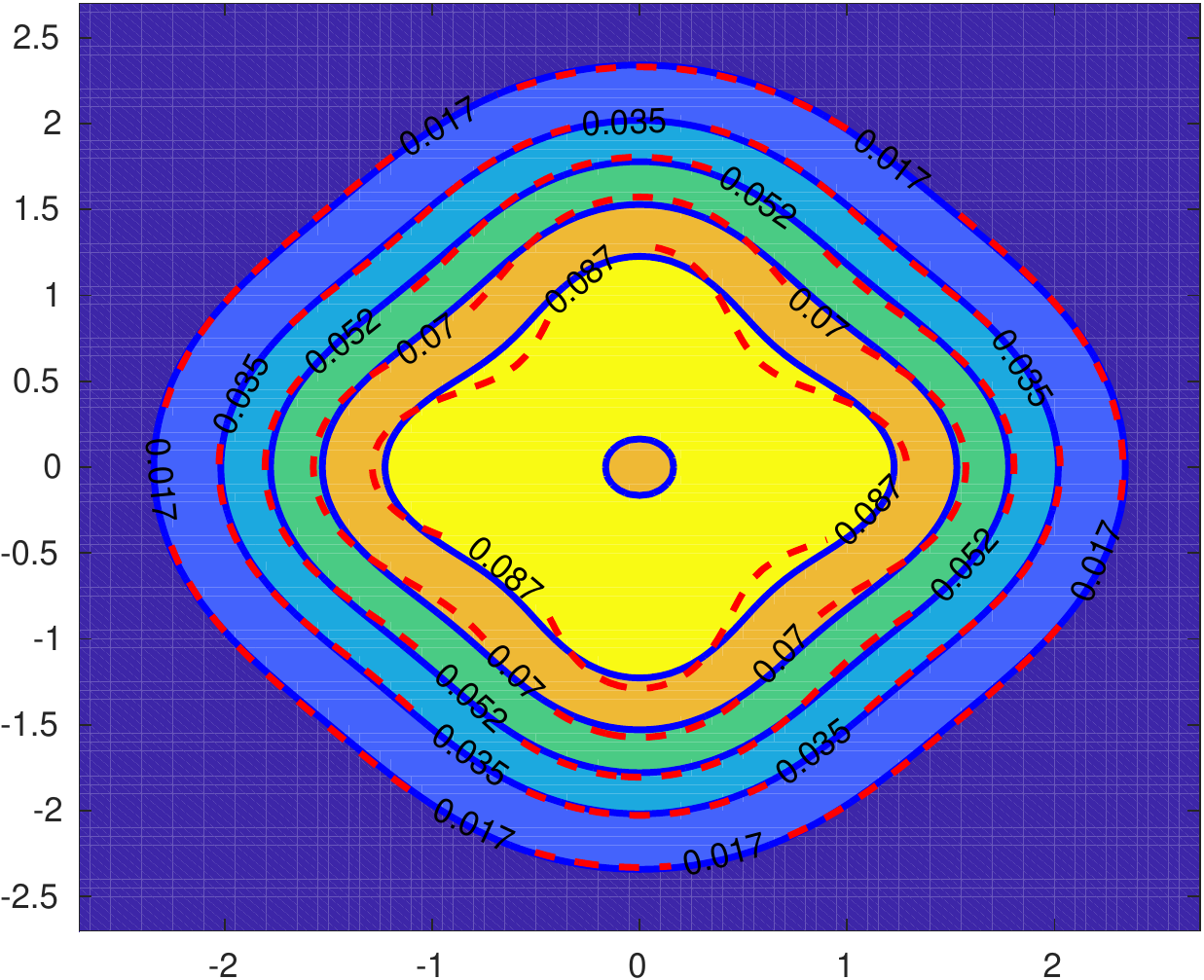}  
  \end{overpic}
}

\caption{Comparison of numerical results using $M_0 = 10$ and
  $M_0 = 20$. The blue contours and the red dashed contours are
  respectively the results for $M_0 = 10$ and $M_0 = 20$.}
\label{fig:ex2_2d_M0=10_40}
\end{figure}

\begin{figure}[!ht]
\centering
\subfigure[$t=0.1$]{%
  \begin{overpic}
    [width=.3\textwidth, clip]{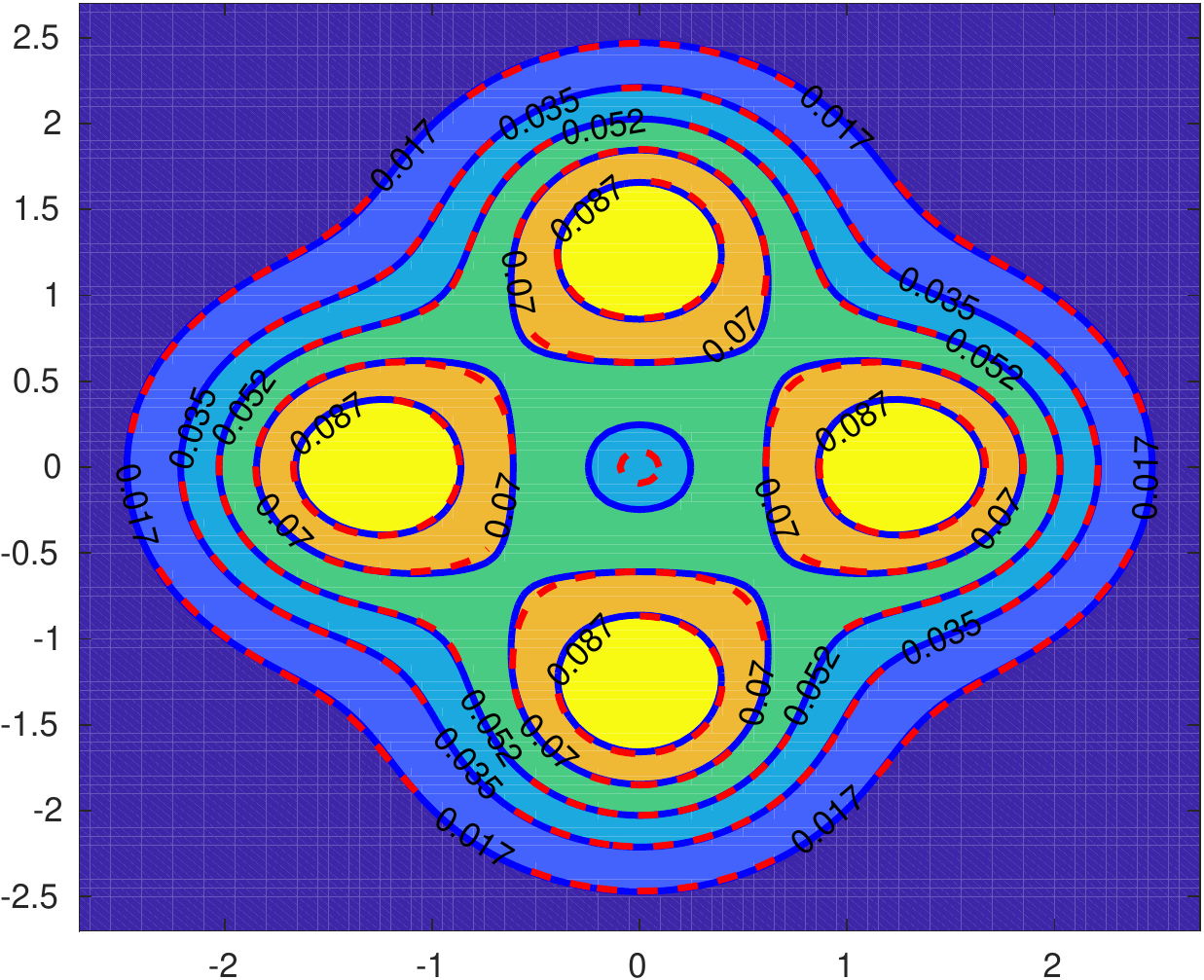}    
  \end{overpic}
} \hfill
\subfigure[$t=0.2$]{%
  \begin{overpic}
    [width=.3\textwidth, clip]{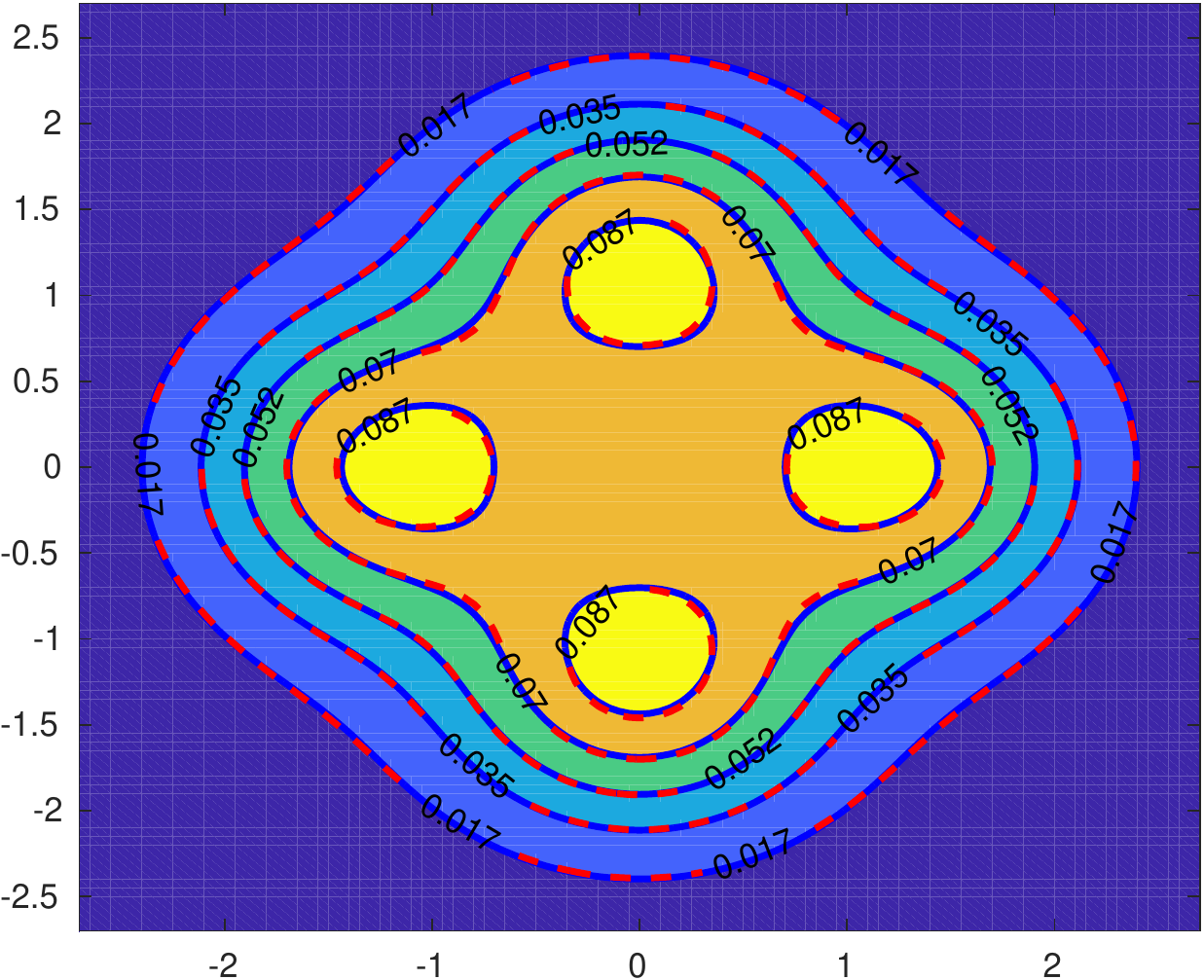}    
  \end{overpic}
} \hfill
\subfigure[$t=0.3$]{%
  \begin{overpic}
    [width=.3\textwidth, clip]{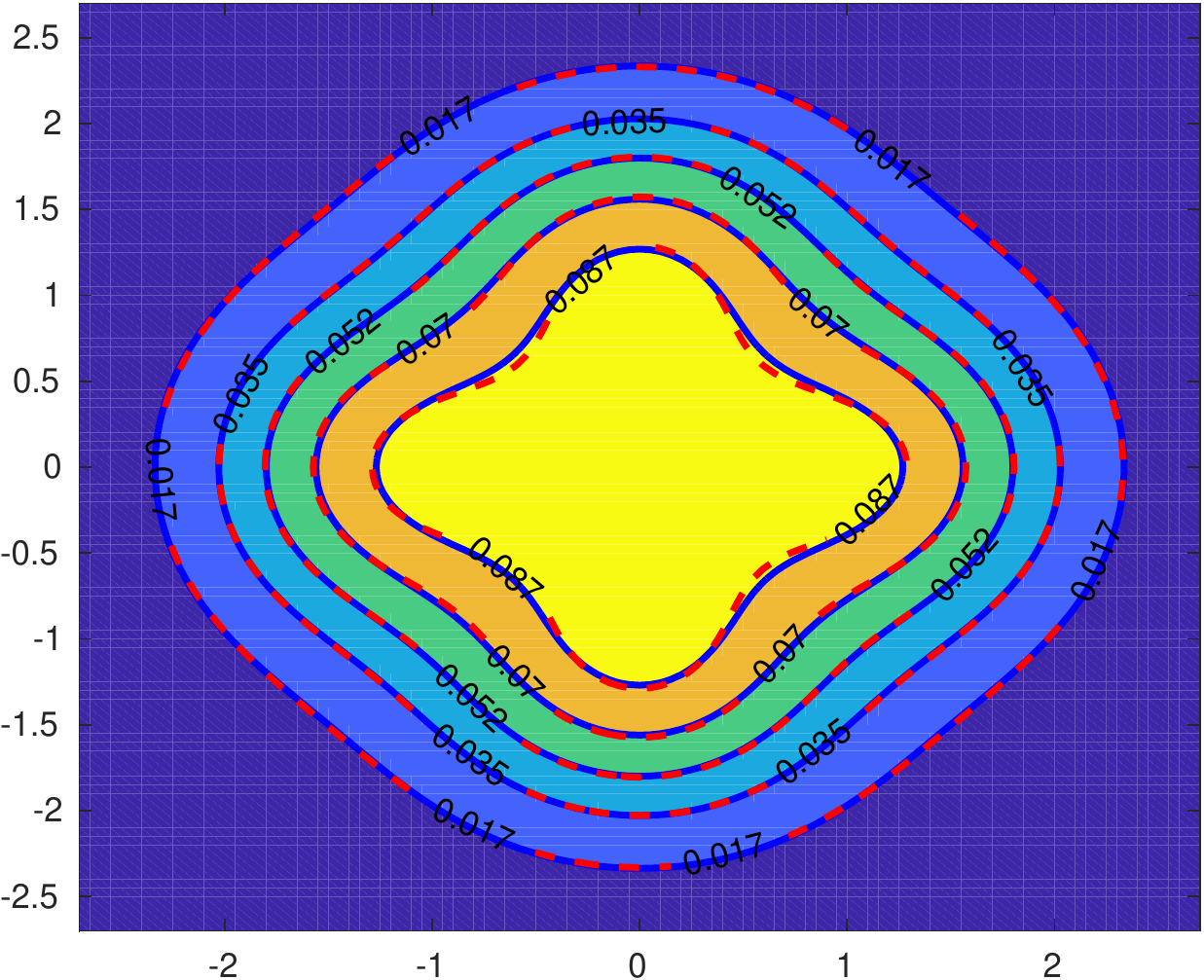}    
  \end{overpic}
  
}
\caption{Comparison of numerical results using $M_0 = 15$ and
  $M_0 = 20$. The blue contours and the red dashed contours are
  respectively the results for $M_0 = 15$ and $M_0 = 20$.}
\label{fig:ex2_2d_M0=15_40}
\end{figure}

Now we consider the evolution of the moments. In this example, the
stress tensor and heat flux satisfy
$\sigma_{11} = \sigma_{22} = -0.5\sigma_{33}$ and
$q_i = 0, i = 1,2,3$. Therefore, we focus only on the evolution of
$\sigma_{11}$, which is plotted in Figure \ref{fig:ex2_sigma11}. It
can be seen that the four tests give almost identical results. Even
for $M_0 = 5$ and $10$, while the distribution functions are not
approximated very well, the evolution of the stress tensor is very
accurate.
\begin{figure}[!ht]
\centering
\includegraphics[width=.4\textwidth]{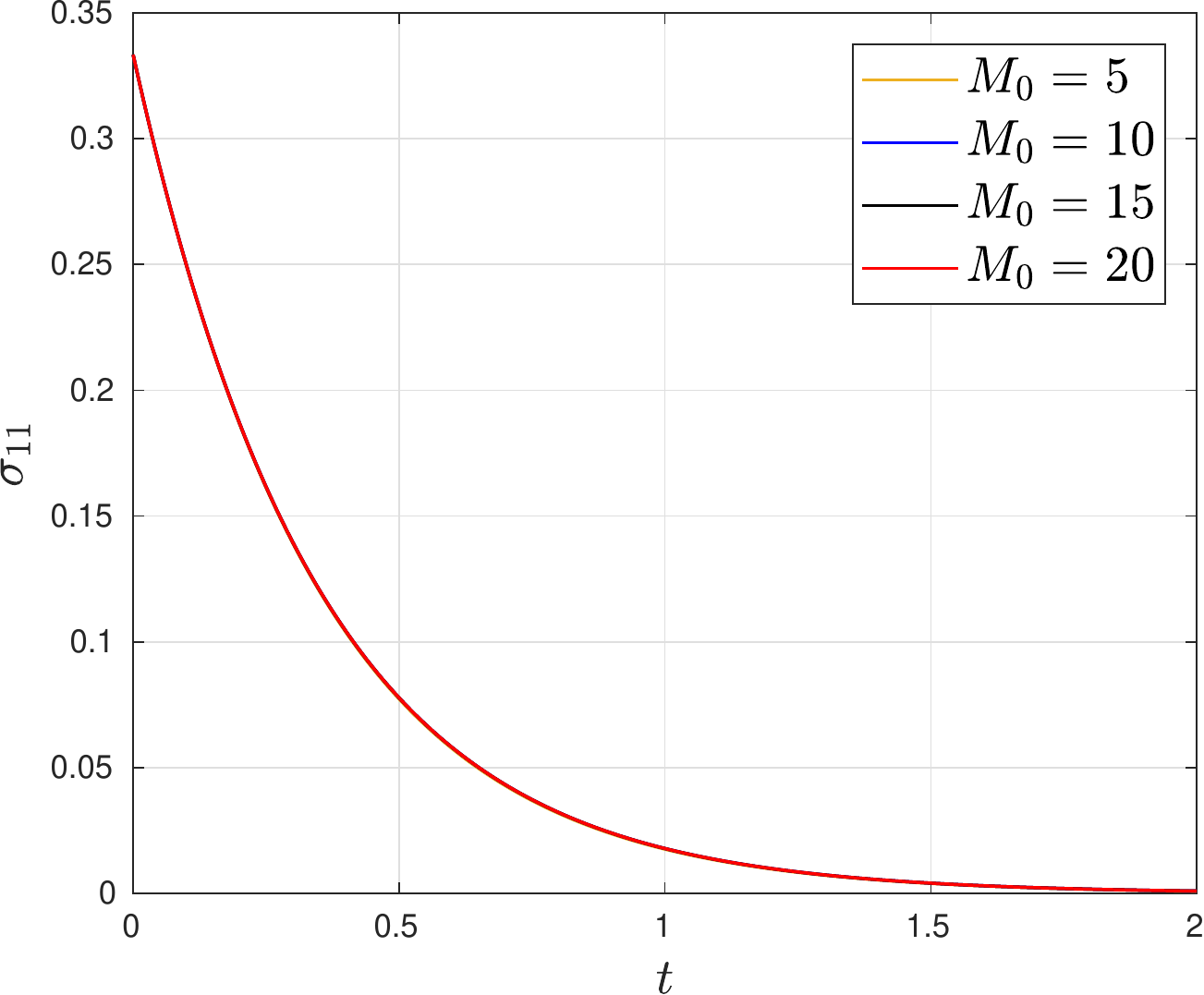}
\caption{Evolution of $\sigma_{11}(t)$. Four lines are on top of each other.}
\label{fig:ex2_sigma11}
\end{figure}

\subsection{Discontinuous initial data}
In this example, we reconsider the problem with the same discontinuous initial
condition as in \cite{QuadraticCol}:
\begin{displaymath}
f^0(\bv) = \left\{ \begin{array}{ll}
  \dfrac{\sqrt[4]{2} (2-\sqrt{2})}{\pi^{3/2}}
    \exp \left( -\dfrac{|\bv|^2}{\sqrt{2}} \right), & \text{if } v_1 > 0, \\[10pt]
  \dfrac{\sqrt[4]{2} (2-\sqrt{2})}{4\pi^{3/2}}
    \exp \left( -\dfrac{|\bv|^2}{2\sqrt{2}} \right), & \text{if } v_1 < 0.
\end{array} \right.
\end{displaymath}
In \cite{QuadraticCol}, the authors used Hermite spectral method to do the
computation up to $M_0 = 15$, which still shows significant difference in the
numerical results compared with $M_0 = 10$. In this paper, we are going to
confirm the reliability of the results obtained with $M_0 = 15$. As in
\cite{QuadraticCol}, we only focus on the evolution of the moments. The
numerical results for the hard potential $\eta = 10$ and soft potential $\eta =
3.1$ with different choices of $M_0$ and $M$ are shown in Figure
\ref{fig:ex3_moments}.  For $\eta = 3.1$, the horizontal axes are the scaled
time $t_s = t / \tau$ with $\tau \approx 2.03942$ as in \cite{QuadraticCol}, so
that the two models have the same mean relaxation time near equilibrium.

Since for the homogeneous Boltzmann equation, the behaviors of stress tensor
and heat flux are the same for any $M \geqslant M_0 \geqslant 3$, we let $M =
M_0 = 5$, $10$, $15$ and $20$, and the results are plotted in Figure
\ref{fig:ex3_moments}. The numerical results for $M_0 = 5$, $10$ and $15$ are
exactly the same as \cite{QuadraticCol}. However, due to the significant
enhancement of computational efficiency, we can get the results for $M_0 = 20$,
and the results are almost the same as those for $M_0 = 15$, which indicates
that they should be very close to the exact solution. The whole pictures of
$\sigma_{11}$ and $\sigma_{22}$ show much clearer converging trend of the
numerical solutions with increasing $M_0$, compared with the numerical results
in \cite{QuadraticCol} where the profiles for $M_0 = 20$ was not present. For
the heat flux $q_1$, not surprisingly, the four results are hardly
distinguishable.

\begin{figure}[!ht]
\centering
\subfigure[$\sigma_{11}(t)$ ($\eta = 10$)]{%
  \includegraphics[width=.4\textwidth]{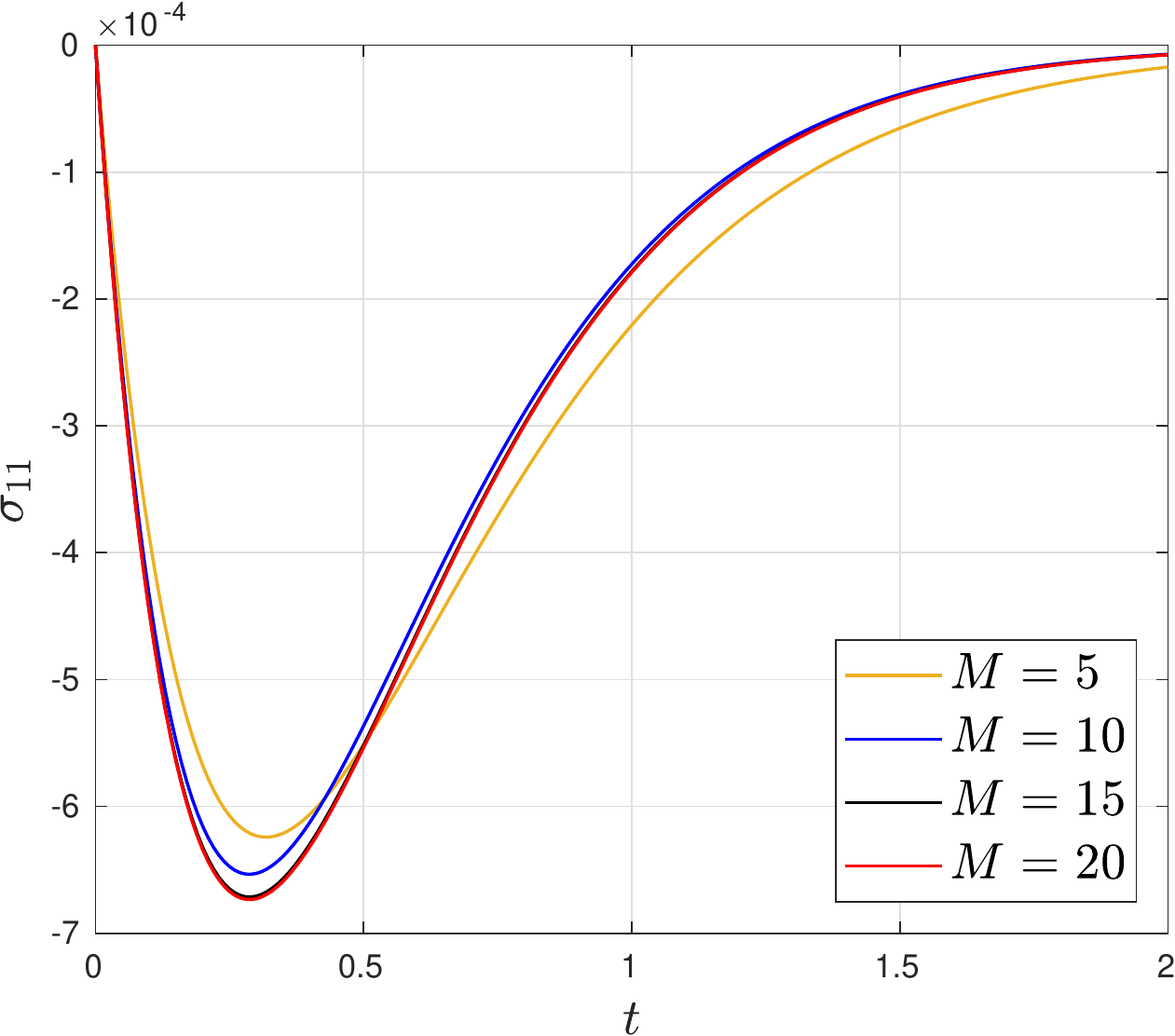}
} \quad 
\subfigure[$\sigma_{11}(t)$ ($\eta = 3.1$)]{%
  \includegraphics[width=.4\textwidth]{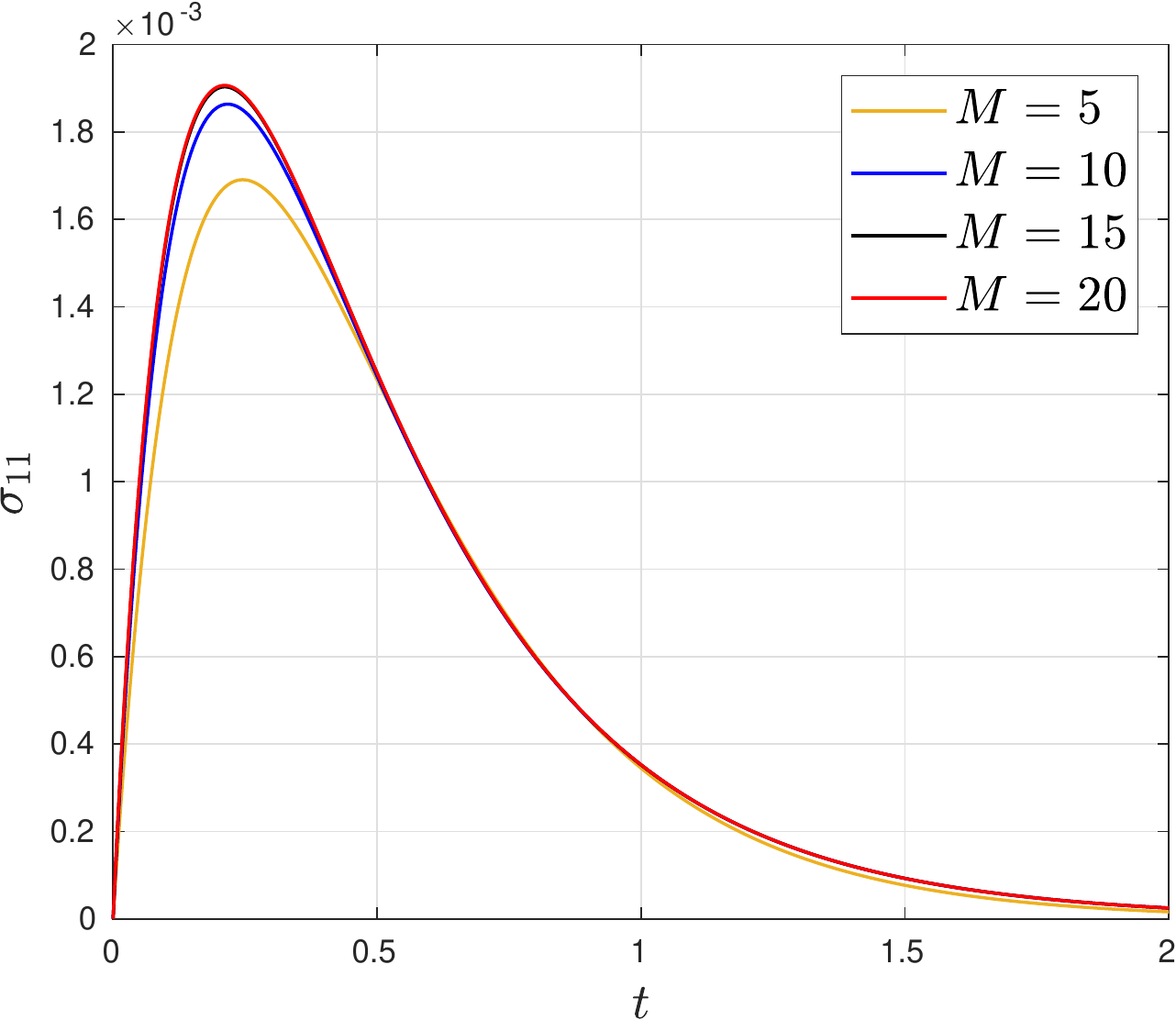}
} \\
\subfigure[$\sigma_{22}(t)$ ($\eta = 10$)]{%
  \includegraphics[width=.4\textwidth]{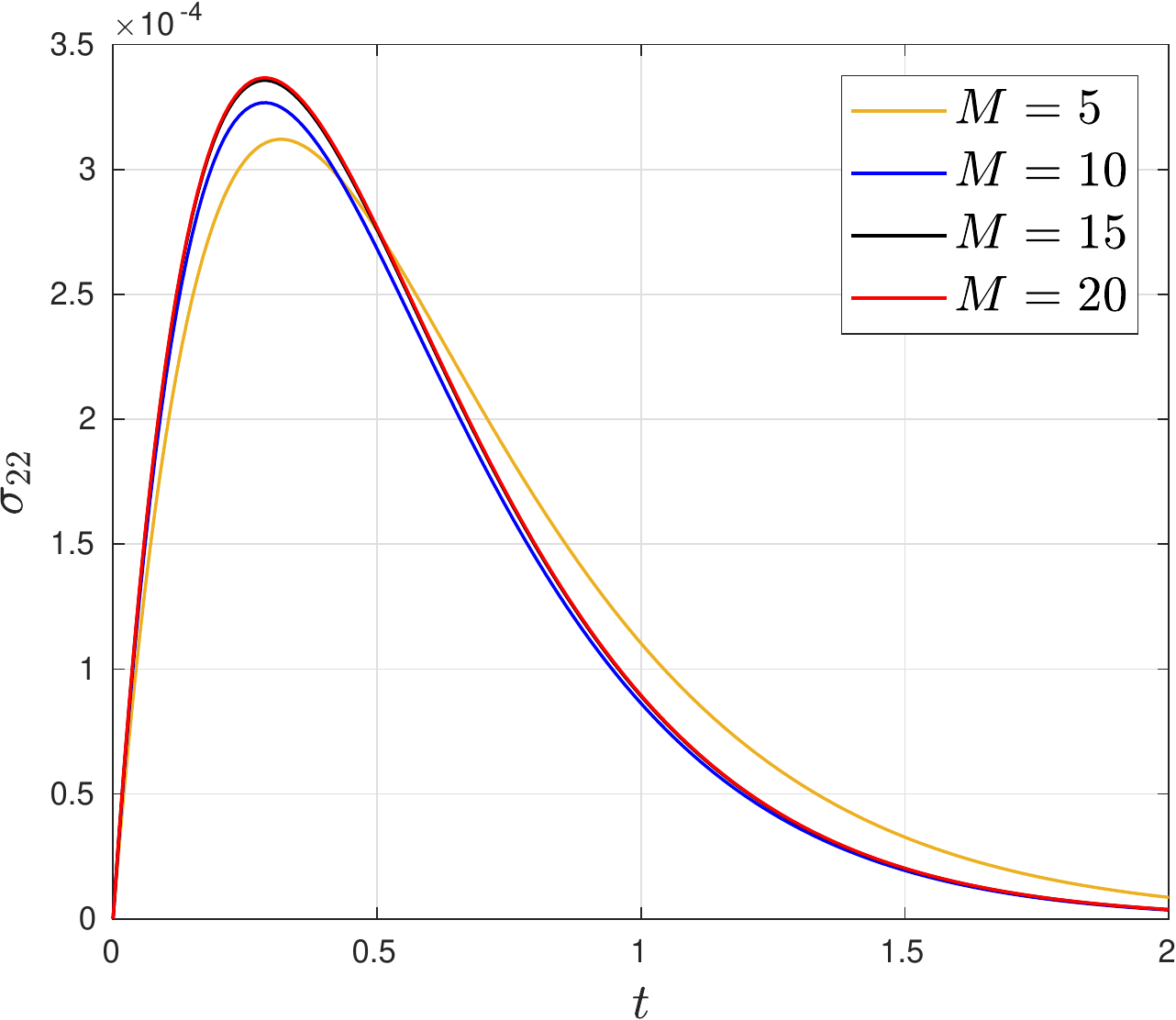}
} \quad 
\subfigure[$\sigma_{22}(t)$ ($\eta = 3.1$)]{%
  \includegraphics[width=.4\textwidth]{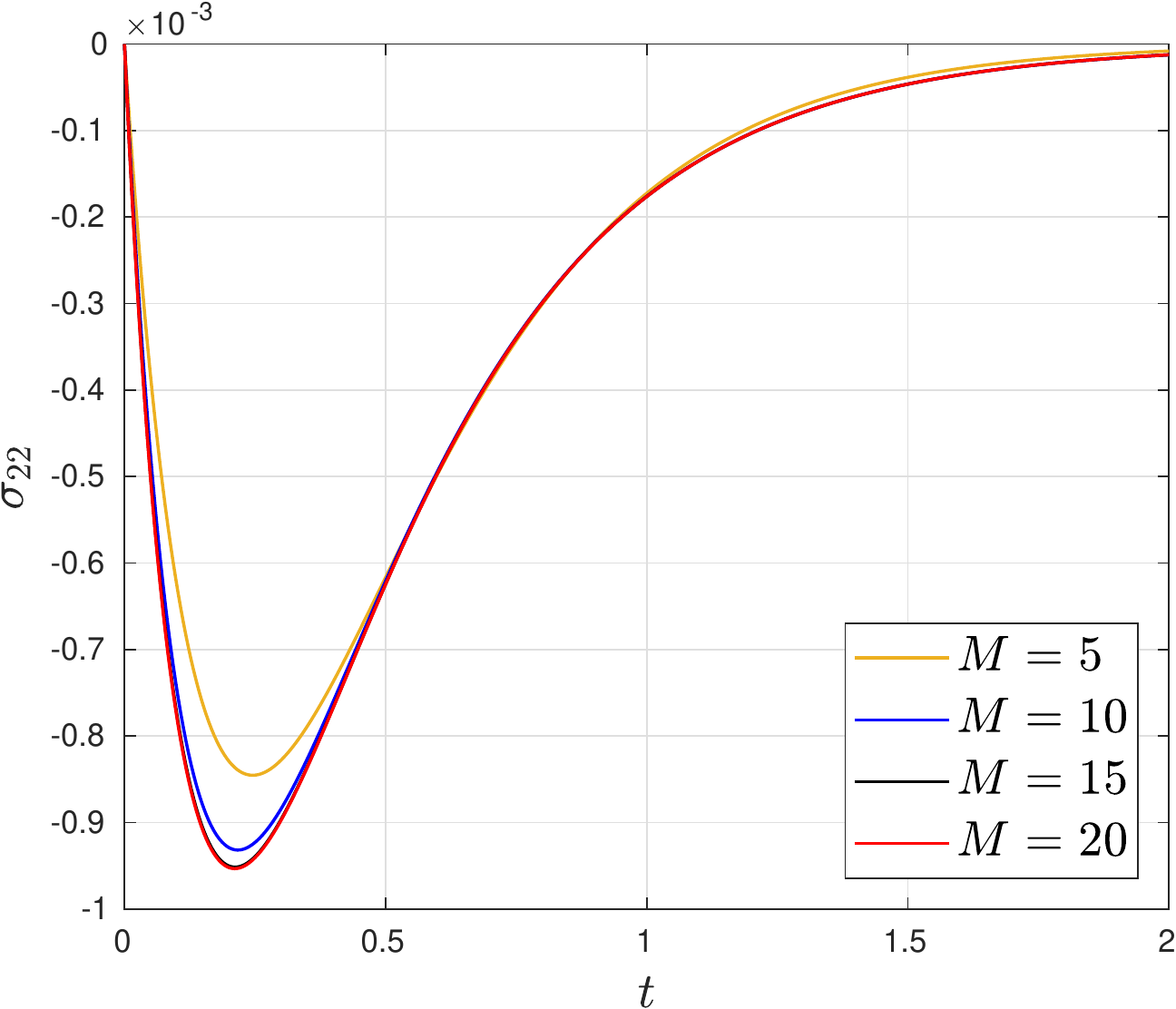}
} \\
\subfigure[$q_1(t)$ ($\eta = 10$)]{%
  \includegraphics[width=.4\textwidth]{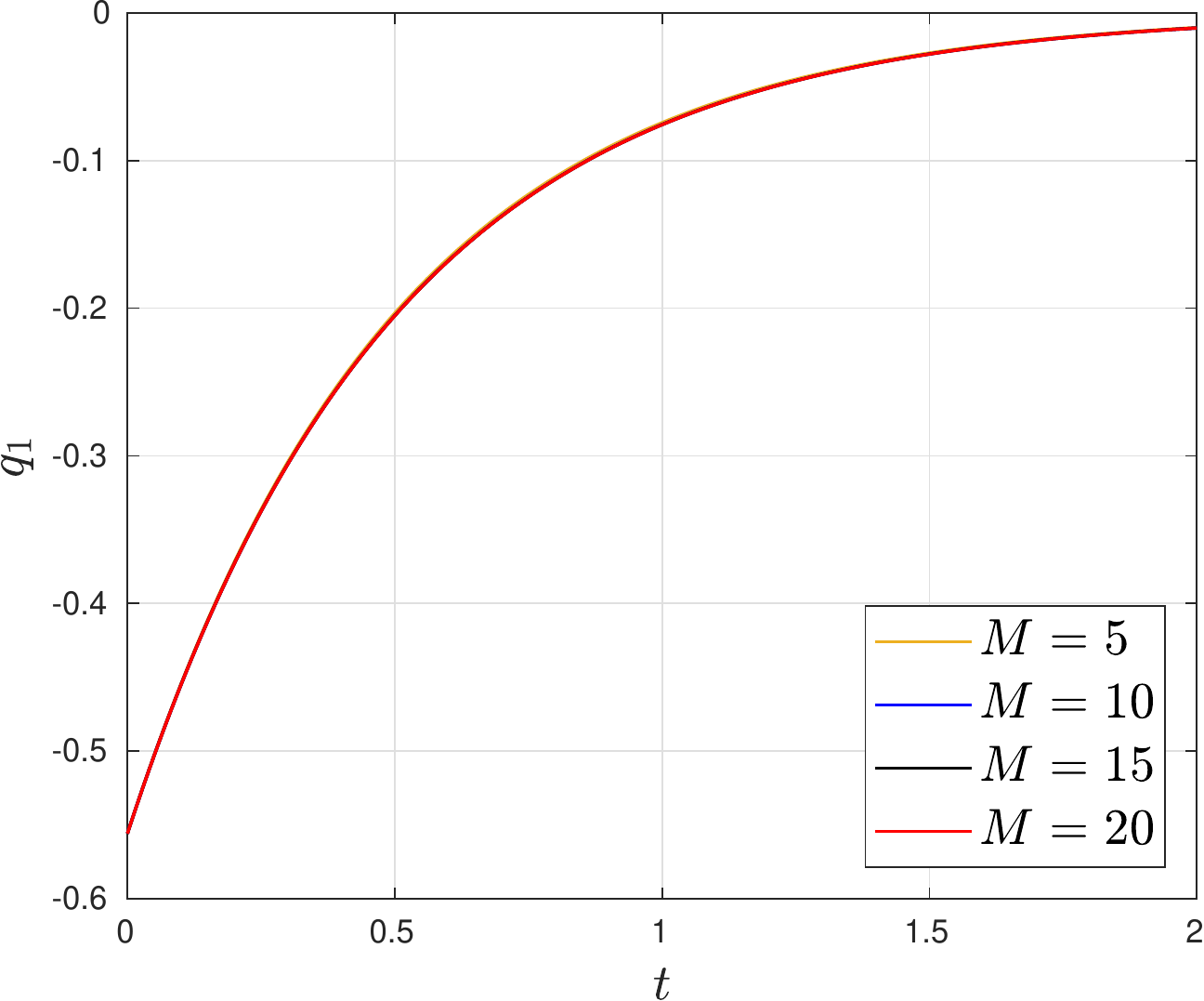}
}\quad 
\subfigure[$q_1(t)$ ($\eta = 3.1$)]{%
  \includegraphics[width=.4\textwidth]{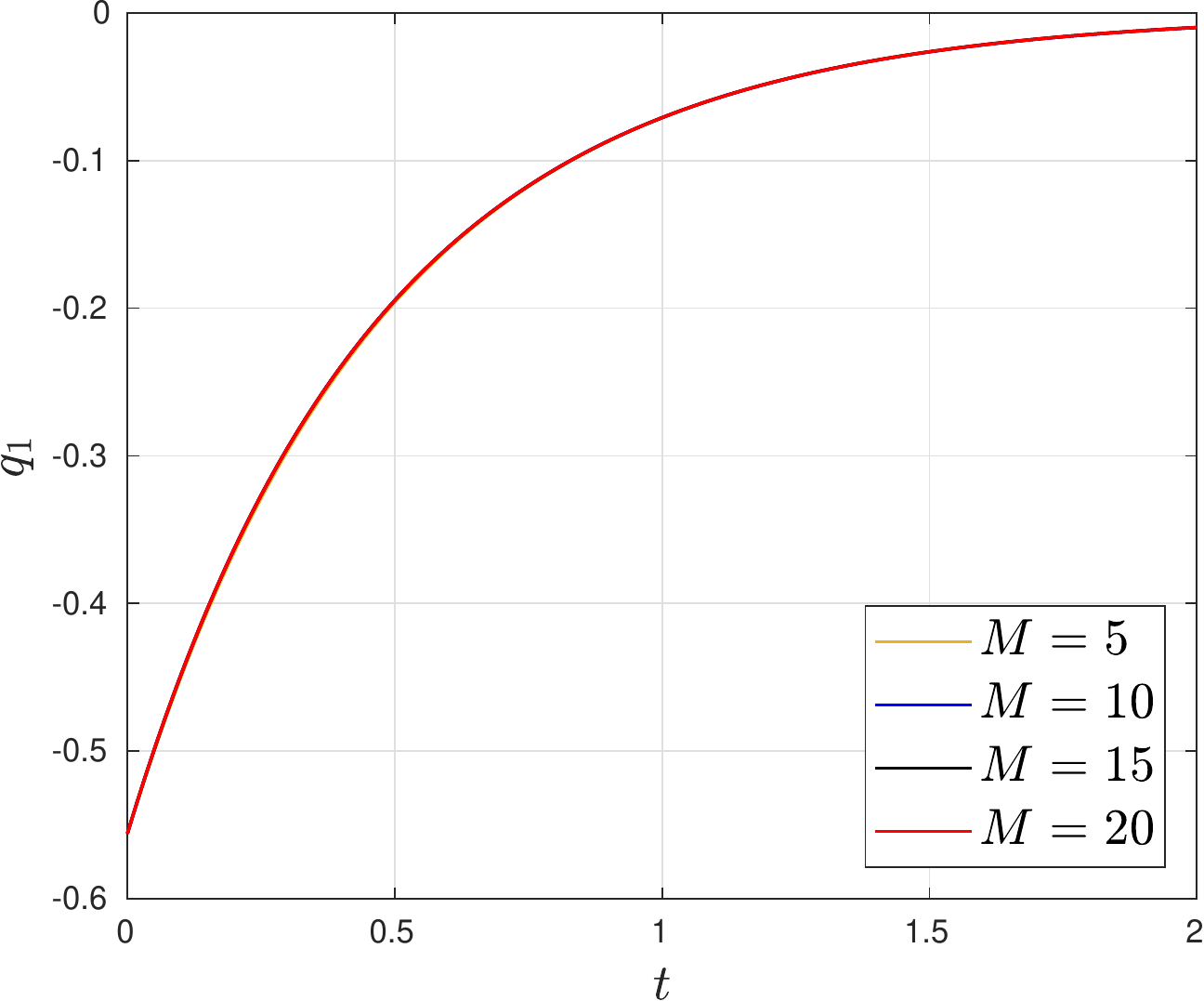}
}
\caption{Evolution of the stress and the heat flux. The left column
  shows the results for $\eta = 10$, and the right column shows the
  results for $\eta = 3.1$. In the right column, the horizontal axes
  are the scaled time.}
\label{fig:ex3_moments}
\end{figure}

Finally, the computational time for one evaluation of the quadratic collision
term under the framework of Burnett series and Hermite series
\cite{QuadraticCol} is plotted in Figure \ref{fig:ex3_time_comparison}. Here,
the number $M$ is fixed as $M = 20$, and $M_0$ increases from $5$ to $20$. It
is clear that the computational cost is greatly reduced by using Burnett basis
functions, especially when $M$ is large.

\begin{figure}[!ht]
\centering
\includegraphics[width=.4\textwidth]{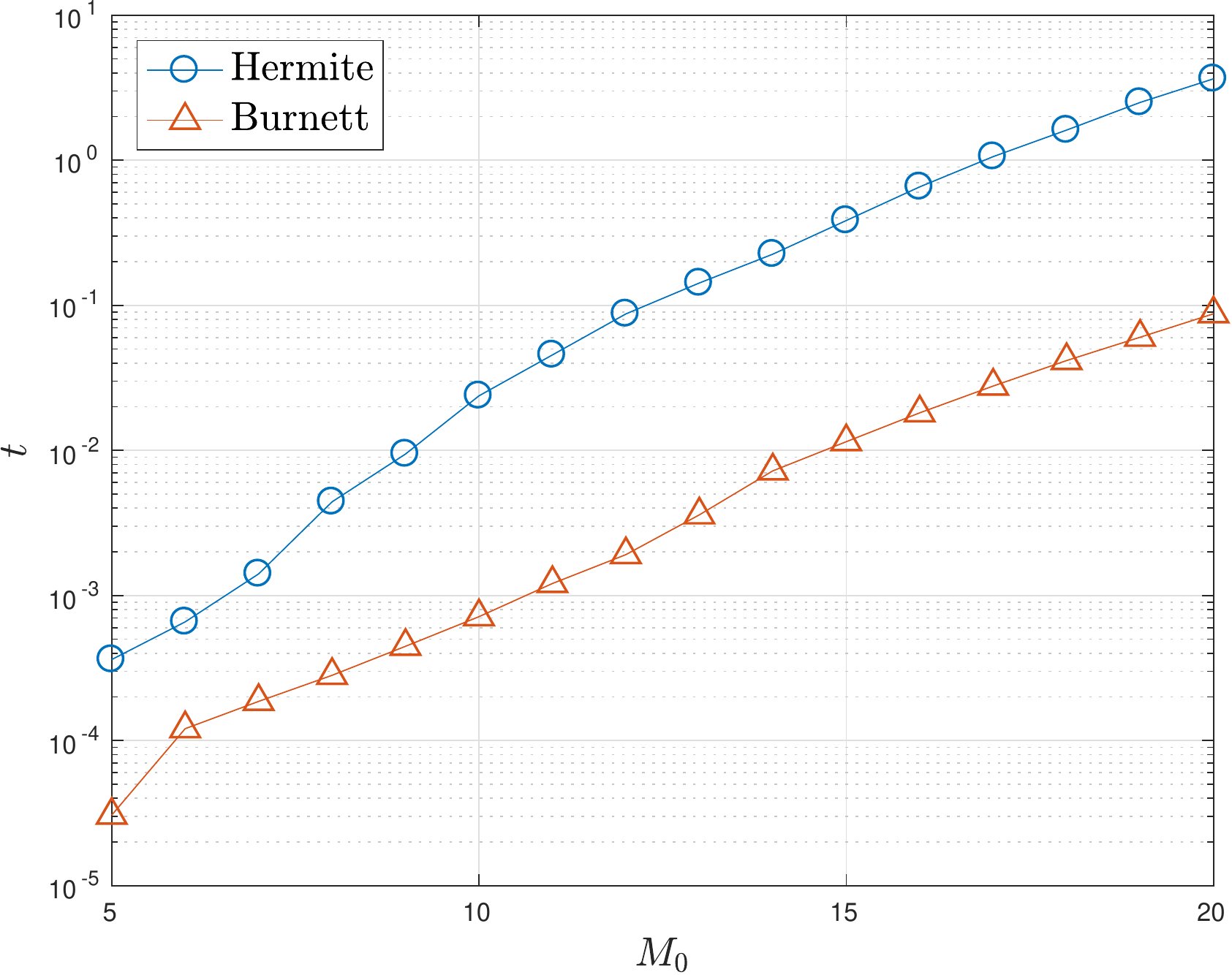}
\caption{Comparison of the computational time for one evaluation of the
  collision operator using the method in this paper and that in
  \cite{QuadraticCol}. $M$ is fixed as $M=20$ and $M_0$ changes from $5$ to
  $20$. The $x$-axis is $M_0$ and the $y$-axis is the logarithm of the
  computational time.}
\label{fig:ex3_time_comparison}
\end{figure}

\section{Proof of theorems}\label{sec:proof}
In this section, we prove the three theorems in Section \ref{sec:general}.
Firstly, we introduce two lemmas as following:
\begin{lemma}\label{lem:rotation}
    Let $\mathbf{R}$ be an $3\times 3$ orthogonal matrix. Define the
    rotation operator $\mathcal{R}$ by
    \begin{displaymath}
        (\mathcal{R}f)(\bv) = f(\mathbf{R} \bv), \qquad
        \forall f: \bbR^3 \rightarrow \mathbb{C}.
    \end{displaymath}
    Then when $\mathcal{Q}[f,g]$ is well-defined for some functions $f$
    and $g$, we have $\mathcal{Q}[f,g](\mathbf{R}\bv) =
    \mathcal{Q}[\mathcal{R}f, \mathcal{R}g](\bv)$.
\end{lemma}
\begin{lemma}\label{lem:Talmi}
    Talmi coefficient $\Talmi{l_1m_1n_1}{l_2m_2n_2}{l_3m_3n_3}{l_4m_4n_4}$
    is zero if 
    \begin{equation}
        l_1+2n_1+l_2+2n_2 \neq l_3+2n_3+l_4+2n_4.
    \end{equation}
\end{lemma}
The first lemma is a well-known result and we are not going to prove it in this
paper. The proof of the second lemma can be find in \cite[page 135-137]{Kumar}.
Now we start to prove the theorems.

\begin{proof}[Proof of Theorem \ref{thm:sparsity}]
For any $\eta \in \bbR$, we define the rotation matrix
\begin{displaymath}
\mathbf{R}_{\eta} = \begin{pmatrix}
  \cos\eta & -\sin\eta & 0 \\ \sin\eta & \cos\eta & 0 \\ 0 & 0 & 1
\end{pmatrix}.
\end{displaymath}
Using spherical coordinates $\bv = (r \sin\theta \cos\phi, r
\sin\theta \sin\phi, r \cos\theta)^T$, one can see that
\begin{displaymath}
\mathbf{R}_{\eta} \bv = (r \sin\theta \cos(\phi + \eta), r
  \sin\theta \sin(\phi + \eta), r \cos\theta)^T.
\end{displaymath}
Therefore
\begin{displaymath}
p_{lmn}(\mathbf{R}_{\eta}\bv) =
  \mathrm{e}^{\mathrm{i} m \eta} p_{lmn}(\bv), \qquad
\varphi_{lmn}(\mathbf{R}_{\eta}\bv) =
  \mathrm{e}^{\mathrm{i} m \eta} \varphi_{lmn}(\bv).
\end{displaymath}
Let $\mathcal{R}_{\eta}$ be the rotation operator such that
$(\mathcal{R}_{\eta} f)(\bv) = f(\mathbf{R}_{\eta} \bv)$. Now we can
rewrite \eqref{eq:A} as
\begin{equation} \label{eq:A_equality}
\begin{split}
A_{lmn}^{l_1 m_1 n_1, l_2 m_2 n_2} &= \int_{\bbR^3}
  \overline{p_{lmn}(\mathbf{R}_{\eta} \bv)}
  \mQ[\varphi_{l_1 m_1 n_1}, \varphi_{l_2 m_2 n_2}](\mathbf{R}_{\eta} \bv)
  \dd \bv \\
&= \int_{\bbR^3} \mathrm{e}^{-\mathrm{i} m \eta} \overline{p_{lmn}(\bv)}
  \mQ[\mathcal{R}_{\eta} \varphi_{l_1 m_1 n_1},
    \mathcal{R}_{\eta} \varphi_{l_2 m_2 n_2}](\bv) \dd \bv \\
&= \mathrm{e}^{\mathrm{i} (m_1+m_2-m) \eta}
  \int_{\bbR^3} \overline{p_{lmn}(\bv)}
  \mQ[\varphi_{l_1 m_1 n_1}, \varphi_{l_2 m_2 n_2}](\bv) \dd \bv \\
&= \mathrm{e}^{\mathrm{i} (m_1+m_2-m) \eta}
  A_{lmn}^{l_1 m_1 n_1, l_2 m_2 n_2}.
\end{split}
\end{equation}
Here we have used the rotational invariance and bilinearity of the
collision operator $\mQ[\cdot, \cdot]$. Note that
\eqref{eq:A_equality} holds for any $\eta$. If $m \neq m_1 + m_2$,
this shows that $A_{lmn}^{l_1 m_1 n_1, l_2 m_2 n_2} = 0$.
\end{proof}

\begin{proof}[Proof of Theorem \ref{thm:sparsityMaxwell}]
    Since $B(g,\chi)=\sigma(\chi)$ is independent of $g$, in
    \eqref{eq:def_Vnn}, the integrals with respect to $g$ and $\chi$
    can be split:
    \begin{equation}
        \begin{split}
            V_{n n'}^l &= \frac{1}{16\sqrt{2} \pi^{5/2}}
            \sqrt{\frac{n!n'!}{\Gamma(n+l+3/2) \Gamma(n'+l+3/2)}}
             \int_0^{\pi} \sigma(\chi) [(2l+1)^2 P_l(\cos \chi) - 1] \dd\chi 
             \times{} \\
            & \qquad \int_0^{+\infty} 
            \left( \frac{g^2}{4} \right)^{l+1}
            L_n^{(l+1/2)} \left( \frac{g^2}{4} \right)
            L_{n'}^{(l+1/2)} \left( \frac{g^2}{4} \right)
            \exp\left( -\frac{g^2}{4} \right) \dd g,
        \end{split}
    \end{equation}
    which vanishes if $n\neq n'$, due to the orthogonality of Laguerre
    polynomials. Hence, using Lemma \ref{lem:Talmi}, we can obtain that
    the summands in \eqref{eq:A_lmn} do not vanish only if
    \begin{equation*}
       n_4=n_4',\quad 
       l_3+2n_3+l_4+2n_4'=l+2n,\quad
       l_3+2n_3+l_4+2n_4=l_1+2n_1+l_2+2n_2.
    \end{equation*}
    Direct simplification yields the conclusion in the theorem.
\end{proof}

\begin{proof}[Proof of Theorem \ref{thm:real}]
    Set $\mathbf{R}=\diag\{1,-1,1\}$, then using the same approach as that in
    the proof of Theorem \ref{thm:sparsity}, one can directly prove this
    theorem.
\end{proof}

\section{Conclusion} 
\label{sec:conclusion}
This work aims to model and simulate the binary collision between gas
molecules under the framework of the Burnett polynomials. The special
sparsity of the coefficients is fully utilized, and we have proposed a
method to compute the coefficients in the spectral expansion with high
accuracy based on the work \cite{QuadraticCol}. Moreover, the data
structure and the implementation of the algorithm are carefully
designed to achieve high numerical efficiency.

In order to provide further flexibility, especially when taking into
account the spatial inhomogeneity, we employ the modelling technique
used in \cite{Cai2015,QuadraticCol}, where the quadratic form is
preserved only for the first few moments. It is validated again that
the method is efficient in capturing the evolution of lower-order
moments. The implementation of the spatially inhomogeneous Boltzmann
equation is in progress.

\section*{Acknowledgements} We would like to thank Prof. Ruo Li at
Peking University, China for the valuable suggestions to this research
project. Zhenning Cai is supported by National University of Singapore
Startup Fund under Grant No. R-146-000-241-133. Yanli Wang is
supported by the National Natural Scientific Foundation of China
(Grant No. 11501042) and Chinese Postdoctoral Science Foundation of
China (2018M631233).

\bibliographystyle{plain}
\bibliography{article}
\end{document}